\documentclass[11pt,letterpaper]{article}

\usepackage[margin=1in]{geometry}
\usepackage{amsmath}
\usepackage{amsthm}
\usepackage{amsfonts}
\usepackage{amssymb}
\usepackage{mathrsfs}
\usepackage{mathtools} %
\usepackage{xfrac}
\usepackage{booktabs}
\usepackage{stmaryrd} %
\usepackage{authblk}
\newcommand*{\email}[1]{\href{mailto:#1}{\nolinkurl{#1}}}

\usepackage{newtxtext}
\usepackage{lipsum} 
\usepackage[dvipsnames,svgnames]{xcolor}
\usepackage[strict]{changepage} 
\usepackage{framed} 
\usepackage{tcolorbox}
\usepackage{todonotes}

\usepackage[backref=page]{hyperref}
\usepackage[nameinlink]{cleveref}
\usepackage[normalem]{ulem}
\usepackage{braket}
\usepackage{caption}
\usepackage{subcaption}
\usepackage[inline]{enumitem}
\usepackage{algorithm}
\usepackage{algpseudocode}
\usepackage{amsthm}
\newtheorem{definition}{Definition}[section]
\newtheorem{theorem}{Theorem}[section]
\newtheorem{lemma}{Lemma}[section]
\newtheorem{corollary}{Corollary}[section]

\newtheorem{claim}{Claim}[section]
\usepackage{enumitem}
\setlist[itemize,1]{leftmargin=1em}
\setlist[itemize,2]{leftmargin=1em}
\setlist[itemize,3]{leftmargin=1em}
\tcbuselibrary{most}

\newcommand{\xx}{\mathbf{x}}

\newcommand{\AND}{\ensuremath{\mathsf{AND}}}
\newcommand{\COPY}{\ensuremath{\mathsf{COPY}}}

\def\PPAD{\textup{\textsf{PPAD}}}
\def\FIXP{\textup{\textsf{FIXP}}}
\def\SDFP/{\textup{\textsf{SDFP}}}

\newcommand{\PureCircuit}{\textup{\textsc{Pure-Circuit}}}
\newcommand{\BFSPureCircuit}{\textup{\textsc{Bounded-Forward-Span Pure-Circuit}}}
\newcommand{\TwoDStrongSperner}{\textup{\textsc{2D Strong Sperner}}}

\newcommand{\NOT}{\ensuremath{\mathsf{NOT}}}
\newcommand{\PURIFY}{\ensuremath{\mathsf{PURIFY}}}
\newcommand{\OR}{\ensuremath{\mathsf{OR}}}

\newcommand{\NOR}{\ensuremath{\mathsf{NOR}}}
\newcommand{\NPURIFY}{\ensuremath{\mathsf{NPURIFY}}}

\begin{document}

\title{Complexity of Autobidding Equilibria in GSP Sponsored Shopping}
\author[*]{Kristoffer Arnsfelt Hansen}
\affil[*]{Aarhus University\\ \email{arnsfelt@cs.au.dk}}
\author[$\ddagger$]{Yuhao Li}
\affil[$\ddagger$]{Columbia University\\ \email{yuhaoli@cs.columbia.edu} }
\author[$\dagger$]{Xinhao Nie}
\affil[$\dagger$]{Aarhus University\\ \email{nie@cs.au.dk}}

\date{}

\maketitle
\thispagestyle{empty}

\begin{abstract}
We revisit autobidding equilibria in GSP sponsored shopping auctions,
recently proposed and studied in~\cite{DLPST26}.
We establish computational upper bounds for both budget constraints
(ABE-B) and return-on-investment constraints (ABE-R): computing an
exact equilibrium is in $\FIXP$, while computing a $\gamma$-approximate
equilibrium is in $\PPAD$ even for inverse-exponentially small $\gamma$.
At the same time, we establish tight inapproximability for both models:
for every fixed $0\leq\gamma<1$, computing a $\gamma$-approximate
equilibrium is $\PPAD$-hard, even when all items are displayed
and all click-through rates are equal.

A key ingredient in our hardness reductions is a new variant of
\PureCircuit\ with bounded forward span, which we prove remains
$\PPAD$-hard. This structural restriction makes
our reduction parameters independent of circuit size, enabling
hardness for every fixed $\gamma<1$. More broadly, our approach
suggests a strategy for strengthening $\PPAD$-hardness results
in other settings, by developing structural insights into the target
problem and new hardness-preserving restrictions on \PureCircuit\
in tandem, with each informing the other and the design of the reduction.
\end{abstract}

\clearpage

\section{Introduction}
As commerce continues to move to digital marketplaces, online shopping platforms have increasingly turned to sponsored product listings to monetize user traffic. Such listings appear on major commerce platforms, including Amazon, Google Shopping, Walmart, Alibaba, and Target. When a user searches for a product, the platform may display a small number of sponsored items in prominent slots alongside organic results. The auctions that determine these allocations, referred to as sponsored shopping auctions, have become an important component of digital commerce.

Sponsored shopping is closely related to the classical sponsored search model, in which advertisers compete for multiple advertising slots and ads are displayed in decreasing order of their bids. The key distinction is that sponsored shopping is inherently a multi-item setting. Whereas in sponsored search an advertiser typically bids for a single link or landing page, in sponsored shopping a retailer may have a large catalog of related products and may wish to promote several of them for the same query. For example, a phone seller may wish to promote several different models for the query ``phone.'' Consequently, a single bidder may win multiple slots in the same auction, each with a different item. The allocation is no longer described solely by matching bidders to positions; rather, it assigns bidder-item pairs to positions, with values depending on both the bidder and the item being displayed.

A canonical auction format in advertising markets is the generalized second-price auction (GSP). In a GSP sponsored shopping auction, the platform ranks all items by their bids and assigns the highest-ranked items to the available slots. The principal departure from the classical sponsored search setting lies in the payment rule. Because a bidder may submit multiple items, the next lower-ranked bid may come from another item owned by the same bidder. GSP \emph{without self-pricing} avoids this artificial self-competition: an item is priced by the next relevant bid in the ranking that belongs to a different bidder. Thus, a bidder's own items cannot set one another's prices. This payment rule is consistent with the auction design described by major advertising platforms such as Google Ads, where eligible keywords from the same account do not compete with one another in the auction~\cite{Goo26}.

This practical setting gives rise to a challenging theoretical model: the GSP payment rule becomes substantially more combinatorial when multiple items may belong to the same bidder. Motivated by the emergence of sponsored shopping, D\"utting, Li, Paes Leme, Spendlove, and Teng~\cite{DLPST26} formally introduced this auction model and studied it through the lens of autobidding, which reflects the structure of modern advertising markets~\cite{ABBBDFGLLMMMMPPPSSSSTWZZZZ24}. In particular, they focused on the canonical \emph{uniform bidding strategy}: rather than submitting a separate bid for each item, each bidder chooses a single multiplier and bids on each item by scaling its value by that multiplier. They adopted autobidding equilibrium~\cite{LT24} as their solution concept and showed that an equilibrium always exists and that every equilibrium achieves at least half of the optimal welfare. These results provide theoretical foundations for the model, showing that its additional combinatorial structure does not preclude meaningful equilibrium and welfare guarantees.

While~\cite{DLPST26} served as a first step toward understanding sponsored shopping, many questions remain open and several directions call for further exploration. Their key technical contribution is a new perturbation technique that yields a clean proof of equilibrium existence even in the more complicated setting of GSP sponsored shopping auctions. However, this technique is inherently nonconstructive and reveals little about the computational structure of the resulting equilibria. In particular, it neither yields an efficient algorithm for computing an exact or approximate autobidding equilibrium nor establishes a meaningful complexity upper bound for the corresponding search problem. This leaves open the central question we study in this paper: what is the computational complexity of autobidding equilibria in GSP sponsored shopping auctions?

Equilibrium computation in autobidding has been extensively studied over the past few years~\cite{CKPSSSW22,CKK24,CKK21,CKSS22,CL25,LT24,FHHH24,YWL26,HYWL26,Liu26b}, all in settings with simple first-price or second-price auctions.
Existing hardness results for pacing and autobidding equilibria
in second-price auctions~\cite{CKK24,LT24,CL25,Liu26b}
rely on markets with many parallel auctions.
By contrast, exact budget-constrained pacing equilibria can be
computed in polynomial time when the number of second-price
auctions is constant~\cite{HYWL26}.
These results suggest that limiting the number of auctions can
restore tractability. Sponsored shopping presents a different
challenge: all competition takes place within a single auction,
but each bidder's multiplier simultaneously controls several items
in a common ranking. Changing one multiplier can therefore alter
both the positions of multiple items and the prices faced by other
bidders. This raises a central question about sponsored shopping:
whether its single-auction structure permits efficient equilibrium
computation, or whether its internal pricing interactions can make
even a coarse approximation intractable.

\subsection{Our results}

In this paper, we provide an arguably comprehensive resolution of these questions by characterizing the computational complexity of autobidding equilibria in GSP-sponsored shopping auctions in both the budget-constrained and return-on-investment (ROI)-constrained settings, for both exact and approximate equilibria.\footnote{We note that \cite{DLPST26} studies sponsored shopping auctions only in the ROI-constrained setting, whereas we address the complexity questions for both the budget-constrained and ROI-constrained settings.}

Our first result is a complexity upper bound for computing an \emph{exact} autobidding equilibrium. Here the relevant complexity class is $\FIXP$, introduced by Etessami and Yannakakis~\cite{EY10} to capture the complexity of
exact fixed point computation for functions represented by algebraic circuits. A canonical $\FIXP$-complete problem is computing an exact Nash equilibrium in a three-player game~\cite{EY10}.
\begin{theorem}
\label{thm:fixp-membership}
\label{thm:Budget-fixp-membership}
\label{thm:ROI-fixp-membership}
For GSP sponsored shopping auctions, computing an exact autobidding equilibrium is in \FIXP\ for each of the two constraint models: budget constraints (ABE-B) and ROI constraints (ABE-R).
\end{theorem}
We then exhibit instances of GSP sponsored shopping auctions, under budget constraints and ROI constraints respectively, in which the autobidding equilibria must be \emph{irrational}. This motivates the natural relaxation to approximate autobidding equilibria.

For our second result, we show that computing an \emph{approximate} autobidding equilibrium is $\PPAD$-complete~\cite{Pap94}, which famously captures for example the complexity of computing two-player Nash equilibria~\cite{DGP09,CDT09}. Our hardness results for both constraint models build on a new
variant of the \PureCircuit\ problem of Deligkas et al.~\cite{DFHM24b}.
We prove that \PureCircuit\ remains $\PPAD$-hard under
\emph{bounded forward span} and exploit this structural restriction
to obtain tight inapproximability.
\begin{theorem}
\label{thm:ppad-results}
\label{thm:Budget-ppad}
\label{thm:ROI-ppad}
\label{thm:Budget-ppad-hard}
\label{thm:ROI-ppad-hard}
For GSP sponsored shopping auctions, the following membership and hardness results hold for both ABE-B and ABE-R. Computing a $\gamma$-approximate autobidding equilibrium is in \PPAD\ when $\gamma$ is inverse-exponential in the input size. Moreover, for every fixed $0\leq\gamma<1$, computing a $\gamma$-approximate autobidding equilibrium is \PPAD-hard.
\end{theorem}

\paragraph{Remark.}
The hardness in \Cref{thm:ppad-results} holds even when the number
of slots equals the number of items and all click-through rates are
one. Every item is then displayed, and each bidder's allocated value
is the same in every allocation; the only quantities that depend on
the ranking are the payments. Thus, computing an ABE-B or an ABE-R
is \PPAD-hard for \emph{every nontrivial constant approximation}
$\gamma<1$, within \emph{a single auction}, with the entire
difficulty residing in the GSP payment rule. The threshold
$\gamma<1$ is tight: for $\gamma=1$ the approximate maximal-pacing
condition is vacuous and the problem becomes trivial. Combined with
the membership result, \Cref{thm:ppad-results} shows that computing
a $\gamma$-approximate equilibrium is \PPAD-complete for every
fixed $0<\gamma<1$. The membership result also covers
inverse-exponentially small approximation parameters.

\paragraph{Our techniques.}
Our upper and lower bounds rest on three technical ideas, each of
which we believe to be of independent interest. We describe them
here at a high level and give a detailed overview in
\Cref{sec:technical-overview}.

\smallskip\noindent
\emph{Quantile coupling for randomized allocations.}
The central obstacle to a fixed point formulation is the randomized
allocation: an item's GSP payment depends on the ranking below it,
including how ties are broken. An existence proof may leave this
distribution implicit, but the circuit must output an explicit,
polynomial-size description that is bid-consistent at every fixed
point. D\"utting et al.~\cite{DLPST26} identify this
global dependence as an obstacle to applying circuit-based
frameworks for \FIXP- and \PPAD-membership.
We overcome it using \emph{quantile coupling}, a classical and
remarkably simple idea from probability theory. We build the ranking by
inserting one group of items at a time, let the circuit compute for
each insertion a marginal distribution supported on legal
positions, and then couple all of these marginals through a single
shared uniform random variable. The resulting joint distribution is
supported on only quadratically many rankings, every ranking in its
support is bid-consistent, and the whole construction is carried
out by a polynomial-size algebraic circuit. To the best of our
knowledge, this is the first use of quantile coupling in establishing
complexity upper bounds for equilibrium computation. It offers a
general recipe for embedding distributions
over combinatorial objects into fixed point constructions, and as a
byproduct it produces equilibria with polynomial-size support in
both constraint models, in which all items of the same bidder with
the same value occupy consecutive positions in every supported allocation.

\smallskip\noindent
\emph{\PPAD-membership for approximate equilibria without rational
exact solutions.}
We exhibit instances of both models in which every exact equilibrium
is irrational, motivating \PPAD-membership for approximate equilibria.
The linear-OPT framework of Filos-Ratsikas, Hansen, H{\o}gh, and
Hollender~\cite{FHHH24} is a powerful tool for
\PPAD-membership, but it was designed for problems with exact
rational solutions, and it is not a priori clear how to use it for
a problem whose exact solutions may be irrational. Our key
observation is that the approximation can be confined to exactly
one operation in our circuit: the multiplication of allocation
probabilities by payments. We replace this product by a one-sided
rounding, which admits a piecewise-linear implementation inside the
framework, and absorb the rounding error into the maximal-pacing
condition, the only condition that the definition of an approximate
equilibrium relaxes anyway. Bid consistency and budget or ROI
feasibility continue to hold exactly at every fixed point, and the
required precision has polynomial encoding length even for
inverse-exponential $\gamma$. This is the first application of the
linear-OPT framework to an approximate equilibrium problem whose
exact solutions may be irrational, and we expect the same strategy
to apply to other problems of this kind.

\smallskip\noindent
\emph{Hardness inside a single ranking, and a new restriction of
\PureCircuit.}
All previous \PPAD-hardness proofs for pacing and autobidding
equilibria~\cite{CKK24,LT24,CL25,Liu26b}
place gadgets in separate auctions, where the competition on each
item can be prescribed independently. Sponsored shopping offers no
such isolation: every item of every gadget enters one common
ranking, so an item introduced for one gate can set the price of an
item in a completely unrelated gate, and the bids that determine
these prices are endogenous equilibrium quantities. Our reductions
turn the distinctive feature of the model, GSP without
self-pricing, into the engine of the hardness proof. A block of $k$
identical items owned by one bidder cannot price one another, so a
single competing bid placed just below the block sets the price of
all $k$ items simultaneously. This amplifies a local price
difference into a $k$-fold payment difference, which overwhelms
both the interference from other gadgets and the slack permitted
by the approximation. Making $k$ a constant depending only on
$\gamma$, rather than on the circuit size, requires controlling how
far the items of a single bidder can spread across value scales.
This is where we depart from all previous uses of \PureCircuit: we
introduce \PureCircuit\ with \emph{bounded forward span}, prove that
it remains \PPAD-hard by revisiting the internal structure of the
original hardness proof, and use it as the source of both
reductions. We discuss this structural approach, which we believe
applies well beyond the present setting, in
\Cref{sec:pure-circuit-discussion}.

\subsection{\PureCircuit\ and structural approaches to \PPAD-hardness}
\label{sec:pure-circuit-discussion}

Many \PPAD-hardness reductions for equilibrium problems have relied on
generalized circuits~\cite{CDT09,Rub18} or
restricted classes of Nash equilibrium problems, including sparse
bimatrix games~\cite{CDT06}, win-lose
games~\cite{AKV05,CTV07}, and two-action
polymatrix games~\cite{CPY17}.
Rubinstein~\cite{Rub18} establishes hardness
for two-action polymatrix games even on bipartite interaction graphs
of degree three.
The choice of source determines which interactions the target
construction must implement and can be decisive for the resulting
hardness theorem.

\PureCircuit, introduced by Deligkas, Fearnley, Hollender, and
Melissourgos~\cite{DFHM24b}, offers a particularly flexible
starting point: its gate constraints have no numerical approximation
parameter. It has supported hardness results across a broad range
of settings. A non-exhaustive list of applications includes
Fisher markets~\cite{DFHM24a},
financial clearing~\cite{IKV26,DHK25},
public goods games~\cite{DH24,DGJNY26},
discrete first-price auctions~\cite{FGHK24},
pacing and throttling equilibria~\cite{CL25,Liu26b},
stochastic games~\cite{HN25}, and sparse win-lose
games~\cite{BFGS26}.
Some of these applications establish tight approximation thresholds,
while others give explicit constant lower bounds.
For sparse win-lose games, Batziou et al.~\cite{BFGS26}
establish hardness at inverse-polynomial precision, with a tight
threshold in sparsity.

Beyond alternative gate bases, the original paper establishes
hardness under several simultaneous structural restrictions:
every variable is an input to exactly one gate, the interaction
graph has total degree at most three, and that graph is
bipartite~\cite[Corollary 2.3]{DFHM24b}.
In some applications, these structural properties are essential to
the reduction itself. In the Fisher-market reduction, bounded out-degree
controls the aggregate demand imposed on a variable-encoding good by
the gadgets that use it~\cite{DFHM24a}.
In discounted perfect-information stochastic games, Hansen and
Nie~\cite{HN25} use bipartiteness to assign the two
sides of the circuit to the two controlling players, enabling the
composition of their gadgets into a two-player game.
In other applications, their role is secondary: they strengthen the
hardness statement with additional restrictions on the hard instances,
rather than enable the reduction itself. For example, the degree bound
in the pacing reduction of Chen and Li~\cite{CL25} yields
hardness even when each bidder values at most four items.

To the best of our knowledge, we are the first to develop a new
structural restriction of \PureCircuit\ beyond those established
in~\cite{DFHM24b}, prove that the resulting variant remains
\PPAD-hard, and exploit this restriction in a reduction.
This restriction is motivated by the interaction between circuit
dependencies and the value scales of our auction construction.
Specifically, we introduce \PureCircuit\ with
\emph{bounded forward span}: its gates are given in an ordering in which
every forward dependency spans only a constant number of gates,
while backward spans are unrestricted. We revisit the internal
hardness construction for \PureCircuit\ to prove that this restriction
preserves \PPAD-hardness (\Cref{thm:ordered-source}).
In our auction reductions, the ordering limits how far an input item
can lie above its bidder's output scale, which guarantees that the total value of a bidder's
input items is at most a constant multiple of its output scale,
independently of circuit size. This allows the construction parameters
to depend only on $\gamma$ and yields hardness for every fixed
$\gamma<1$.
Bounded forward span may also be useful in reductions to other
position auctions where a common ranking couples gadgets placed at
different value scales.\footnote{Applications in other settings may
benefit from combining bounded forward span with alternative gate
bases, such as those in~\cite[Corollary 2.2]{DFHM24b}.
Combining bounded forward span with the additional structural
restrictions in~\cite[Corollary 2.3]{DFHM24b} would require
a separate hardness argument.}

We hope that this approach provides a useful methodology for
strengthening existing \PPAD-hardness results and establishing new
ones: identify the structural requirements of the target problem,
develop hardness-preserving restrictions of \PureCircuit\ that meet
those requirements, and refine the source construction and the target
reduction in tandem.
Similar ideas may also be explored for generalized circuits~\cite{CDT09,Rub18} and
threshold games~\cite{PP23}, by developing
structural restrictions suited to their applications.

\section{Preliminaries}
We recall the setting of a position auction where the bidders can have multiple items: 
There are $n$ bidders and $K$ slots. Each bidder $i$ has $m_i$ items (or ads), 
and the values are given in sorted order $v_{i,1}\geq\cdots\geq v_{i,m_i}> 0$. 
We let $M\coloneqq \sum_{i \in [n]} m_i$ denote the total number of items.
As usual, the slots in a given position auction have decreasing importance to bidders, 
which is modeled by the click-through rates $c_1\geq c_2\geq\cdots\geq c_K>0$. 
This means if item $j$ of bidder $i$ is allocated to slot $k$, then bidder $i$ obtains value 
$c_k v_{ij}$ for this item. Without loss of generality, we assume that $c_1\leq 1$ and 
extend the CTR sequence by setting $c_p \coloneqq 0$ for $K < p \leq M$. Thus, positions beyond $K$ 
are included in the ordering but receive no clicks.

In this work, we are mainly interested in the generalized second-price (GSP) auction. 
A key feature of GSP in our multi-item setting is that the payment rule should satisfy 
\emph{no self-pricing}. 

\subsection{GSP without self-pricing}
We start by describing the GSP allocation rule. 
Each agent $i$ will submit the bids $b_{i,1},\cdots,b_{i,m_i}$ for their items. 
Let $\mathbf{b}=(b_\ell)_{\ell \in [M]}$ denote the vector of bids from all bidders. 
The GSP mechanism will sort all items by their bids and allocate the item 
with the $k$-th highest bid to the $k$-th slot. In the event of tied bids, we are allowed to use a randomized tie-breaking rule. 

Formally, let $\Omega$ be the set of all permutations over the items $[M]$. Let $a\in\Omega$ be an allocation, where $a(k)=(i,j)$ 
means that the item at rank $k$ is bidder $i$'s $j$-th item. An allocation is \emph{valid under bids $\mathbf{b}$}
if the items are ordered in non-increasing order of their bids:
\[
b_{a(1)}\geq b_{a(2)}\geq\cdots \geq b_{a(M)}.
\]
Naturally, the value of bidder $i$ under an allocation is 
\[
V_i(a)\coloneqq \sum_{k\in[K]: a(k)=(i,j)} c_k\cdot v_{ij}. 
\]

Next, we describe the generalized second-price (GSP) payment rule \emph{without self-pricing} 
for multi-item, multi-slot position auctions. The basic intuition behind GSP is unchanged: 
each item pays according to the next-highest bid. In our setting, however, it is natural to come with a refinement: 
an item is priced only against items submitted by other bidders.

Fix bids $\mathbf b$ and a valid allocation $a\in\Omega$ under $\mathbf b$.
For each item $(i,j)$, let $k$ denote its position in the global order $a$, i.e.,
$a(k)=(i,j)$.
The GSP \emph{per-click price} for item $(i,j)$ is defined as
\[
\tau_{ij}(a,\mathbf b)
\coloneqq
\max\Bigl\{
b_{i'j'}
\mid
(i',j') = a(k'),\ k'>k,\ i'\neq i
\Bigr\},
\]
with the convention that the maximum of an empty set is zero. Namely, each allocated item should pay the next-highest bid that is \emph{not} among their bids, thereby ruling out self-pricing. Letting $k$ be the position of item $(i,j)$ in $a$, the GSP payment for item $(i,j)$ is defined as 
\[
p_{ij}(a,\mathbf b)\coloneqq c_k\cdot \tau_{ij}(a,\mathbf b).
\]
The total payment of bidder $i$ is then given by
\[
P_i(a,\mathbf b)
\coloneqq
\sum_{j=1}^{m_i} p_{ij}(a,\mathbf b).
\]

The existence of autobidding equilibria often hinges on tie-breaking rules~\cite{CKSS22,LT24}. Our more complex sponsored-shopping setting is no exception: D\"utting, Li, Paes Leme, Spendlove, and Teng~\cite{DLPST26} therefore introduced tie-breaking rules for this setting, which we recall below.

Let $\pi$ be a randomized allocation, which is a distribution over $\Omega$. We say $\pi$ is \emph{valid under bids $\mathbf{b}$} if every allocation in its support is valid under $\mathbf{b}$. So a valid randomized allocation is essentially a \emph{tie-breaking rule}. Given $\mathbf{b}$ and a valid $\pi$, the expected value of bidder $i$ and the expected payment of bidder $i$ are respectively
\[
V_i(\pi)\coloneqq \mathbb{E}_{a\sim \pi}[V_i(a)]\quad\text{and}\quad P_i(\pi,\mathbf b)
\coloneqq
\mathbb{E}_{a\sim \pi}\left[ P_i(a,\mathbf b)\right]. 
\]

\subsection{Autobidding and equilibria}

In this section, we recall autobidding and the solution concepts of autobidding equilibria. 

In modern ad marketplaces, autobidding is the prevalent form to manage advertising campaigns. The autobidding behavior is modeled by the \emph{uniform bidding} strategy, defined as follows:
\begin{definition}[Uniform Bidding]
    We say an agent $i$ adopts a uniform bidding strategy if they select a multiplier and bid the true value of each item, scaled by this single multiplier. More formally, the agent $i$ selects $\alpha_i\geq 0$ as the multiplier, and bids $\alpha_iv_{ij}$ for their $j$-th item. 
\end{definition}

A growing body of work (see e.g. \cite{ABM19,CKPSSSW22,CKSS22,CKK24,LT24,CL25,DLPST26}) studies equilibrium behavior in settings where all agents adopt autobidding strategies. In this model, each autobidder is subject to certain constraints, with budget constraints and return-on-investment (ROI) constraints being the most commonly studied. Correspondingly, the literature has proposed solution concepts of autobidding equilibria under budget constraints and ROI constraints. We discuss them for our setting below.

When the autobidders are under budget constraints, each agent $i$ has a budget $B_i>0$. The autobidding equilibria under budget constraint (ABE-B) and its relaxed variant are defined as follows.

\begin{definition}[Autobidding Equilibria under Budget Constraint]\label{definition: pacing equilibrium budget}
\sloppy    Given a position auction $G = (n, (m_i), K, (v_{ij}), (c_k), (B_i))$, a pair $(\alpha, \pi)$, where $\alpha=(\alpha_i)\in [0,1]^n$ and $\pi$ is a distribution over all allocations $a\in\Omega$, 
  is an \emph{autobidding equilibrium under budget constraint} of $G$ if the following conditions hold:
\begin{itemize}[leftmargin=2em]
    \item[(a)] \emph{Bid-consistent allocation.}
    Every allocation in the support of $\pi$ is valid under the induced bids
    $\mathbf b$, i.e., items are ordered in non-increasing order of bids.

    \item[(b)] \emph{Budget feasibility.}
    For every bidder $i\in[n]$, the expected payment should not exceed their budget: $P_i(\pi,\mathbf b)\leq B_i.$
    That is, each bidder satisfies their budget constraint.

    \item[(c)] \emph{Maximal pacing.}
    For every bidder $i\in[n]$, if the budget constraint is slack $P_i(\pi,\mathbf b)<B_i$,
    then bidder $i$ is maximally paced, i.e., $\alpha_i = 1$.
\end{itemize}
\end{definition}

For defining notions of approximate autobidding equilibria under budget constraint, one may in principle consider 
relaxing any subset of the conditions (a), (b), and~(c). For the sponsored search setting, \cite{CKK24}
relaxed the corresponding conditions~(a) and~(c). Arguably, relaxing condition~(c) yields the most reasonable notion of 
approximation.

\begin{definition}[Approximate Autobidding Equilibria under Budget Constraint]
    \label{definition: approximate pacing equilibrium budget}
In the setting of \Cref{definition: pacing equilibrium budget} we say $(\alpha,\pi)$ is an $\gamma$-approximate ABE-B,
 for given $0\leq \gamma < 1$, if conditions~(a) and~(b) of \Cref{definition: pacing equilibrium budget} 
 as well as the relaxed conditions hold:
\begin{itemize}[leftmargin=2.5em]
    \item[(c')] \emph{Approximately maximal pacing.}
    For every bidder $i\in[n]$, if the budget constraint is slack by a \(\gamma\)-margin $P_i(\pi,\mathbf b)<(1-\gamma)B_i$,
    then bidder $i$ is maximally paced, i.e., $\alpha_i = 1$.
\end{itemize}
\end{definition}

When the autobidders are under the ROI constraints, the autobidding equilibria under ROI constraints (ABE-R) and its relaxed variant are defined as follows.

\begin{definition}[Autobidding Equilibria under ROI Constraint]\label{definition: pacing equilibrium ROI}
\sloppy Given a position auction $G = (n, (m_i), K, (v_{ij}), (c_k))$, a pair $(\alpha, \pi)$, where $\alpha=(\alpha_i)\in [1,A]^n$ and $\pi$ is a distribution over all allocations $a\in\Omega$, 
  is an \emph{autobidding equilibrium under ROI constraint} of $G$ if the following conditions hold:
\begin{itemize}[leftmargin=2em]
    \item[(a)] \emph{Bid-consistent allocation.}
    Every allocation in the support of $\pi$ is valid under the induced bids
    $\mathbf b$, i.e., items are ordered in non-increasing order of bids.

    \item[(b)] \emph{ROI feasibility.}
    For every bidder $i\in[n]$, the expected value obtained under $\pi$
    is at least the expected payment: $
    V_i(\pi) \ge P_i(\pi,\mathbf b).$
    That is, each bidder satisfies their return-on-investment (ROI) constraint.

    \item[(c)] \emph{Maximal pacing.}
    For every bidder $i\in[n]$, if the ROI constraint is slack, $V_i(\pi) > P_i(\pi,\mathbf b)$,
    then bidder $i$ is maximally paced, i.e., $\alpha_i = A$.
\end{itemize}

\end{definition}

For defining notions of approximate autobidding equilibria under ROI constraint, one may consider relaxing any subset of 
the conditions (a), (b), and~(c). For the sponsored search setting, \cite{LT24} relaxed all of the corresponding conditions.
Again, relaxing condition~(c) only arguably yields the most reasonable notion of approximation.

\begin{definition}[Approximate Autobidding Equilibria under ROI Constraint]
    \label{definition: approximate pacing equilibrium ROI}
In the setting of \Cref{definition: pacing equilibrium ROI} we say $(\alpha,\pi)$ is an $\gamma$-approximate ABE-R, 
for given $0\leq \gamma< 1$, if conditions~(a) and~(b) of \Cref{definition: pacing equilibrium ROI} as well as the following relaxed condition hold:
\begin{itemize}[leftmargin=2.5em]   
    \item[(c')] \emph{Approximately maximal pacing.}
    For every bidder $i\in[n]$, if the ROI constraint is slack by a $\gamma$-margin, $(1 - \gamma)V_i(\pi) > P_i(\pi,\mathbf b)$,
    then bidder $i$ is maximally paced, i.e., $\alpha_i = A$.
\end{itemize}
\end{definition}

Prior work~\cite{LT24,DLPST26} uses a sufficiently
large multiplier cap $A$, leaving its precise choice and the
relationship between equilibria at different caps implicit.
We make this convention precise through the natural choice
$A=A_0\coloneqq 1+\frac{\sum_{i,j}v_{ij}}{\min_{i,j}v_{ij}}$.
With this choice, the cap serves to give a bounded representation
of equilibrium behavior: every exact or $\gamma$-approximate ABE-R
at $A_0$ represents a family of equilibria for all larger caps,
with the same allocation distribution, payments, and approximation
parameter.\footnote{
To extend an equilibrium to a larger cap, raise only the multipliers
equal to $A_0$ to the new cap and keep the allocation distribution
unchanged. The justification is given in \Cref{app:multiplier-cap}.}
Moreover, $A_0$ is computable in polynomial time and has polynomial
encoding length.

\section{Technical overview}
\label{sec:technical-overview}

\subsection{Membership}

\paragraph{\FIXP-membership.}
For the $\FIXP$ construction in \Cref{sec:fixp-membership}, we first group together each bidder's
items with the same value and insert these groups one at a time.
At each step, a sorting network computes the numerical order of the
previous bids. Comparisons with the new group's bid determine its legal
insertion positions, including all positions permitted by ties.
Using Heaviside pseudogates from the framework
of~\cite{FHHH23,FHHH24}, we compute
a probability distribution supported on these positions at every
fixed point. 

We combine these insertion marginals using quantile coupling.
The circuit sorts all entries of their cumulative probability vectors.
For each interval of positive length between consecutive entries, its
midpoint selects an insertion position from every marginal, producing
a deterministic group order. Assigning the probability equal
to the interval length recovers all prescribed insertion marginals.
If there are $T$ groups, there are only $O(T^2)$ cumulative entries,
so the resulting distribution has polynomial-size support.

The circuit then tracks
how later insertions shift earlier groups, and uses group
multiplicities and a fixed order within each group to recover
the individual item positions.
These positions determine the next bid from another bidder and the
CTR of each item, allowing us to compute each supported allocation's
GSP payments and the bidders' expected payments and valuations.
We complete the circuit by increasing multipliers when their
constraints are slack and decreasing them when their constraints
are violated.

\paragraph{\PPAD-membership.}
Exact equilibria can be irrational. \Cref{sec:ppad-membership}
gives small rational instances of both models in which every
equilibrium is irrational. Both examples reduce to a quadratic
equation in which a tie-breaking probability is multiplied by a
multiplier. We therefore seek a \PPAD\ upper bound for approximate
equilibria.

For $\PPAD$-membership in \Cref{sec:ppad-membership}, we preserve the allocation construction and
replace the remaining nonlinear payment computations by
discretizations. The allocation construction already uses piecewise-linear
pseudo-circuits. Moreover, products of the form $yH(x)$ admit PL
pseudo-circuit implementations by
\cite[Lemma~3.1]{FHHH24}. It therefore lets us compute each supported allocation's
payment exactly within this framework. The remaining obstacle is
multiplication between probability weights and variable payment
amounts. We use the piecewise-constant correspondence represented
by an identity Boolean circuit to round downward, applying
\cite[Proposition~3.3]{FHHH24} to implement the entire
scaled correspondence. Its circuit size is polynomial in the number
of bits specifying the discretization precision.

The two constraint models require different error bounds. Under budget
constraints, we discretize each allocation's total payment and scale
its error by the allocation probability. Since these probabilities
sum to one, we can ensure
$0\leq P_i-\widehat P_i\leq\gamma B_i$.
Under ROI constraints, the relevant valuation may be arbitrarily
small. We first compute probability-weighted CTRs exactly, then
discretize each group's per-click price using the multiplier cap $A$
times the group's value as an upper bound.
The resulting errors sum to at most $\gamma V_i$, where
$V_i$ is also computed exactly by a PL pseudo-circuit.
Using $(1-\gamma)B_i-\widehat P_i$ or
$(1-\gamma)V_i-\widehat P_i$ in the multiplier updates then preserves
budget or ROI feasibility exactly and enforces approximate maximal
pacing at every fixed point. The required precision has polynomial
encoding length even for inverse-exponential $\gamma$.

\subsection{Hardness}

\paragraph{Auction construction.}
The hardness reduction for budget constraints starts from
\PureCircuit, which asks for an assignment $\xx:V\to\{0,1,\bot\}$
satisfying its gate constraints.
We use a $\NOT$ gate to illustrate the construction, leaving
the spacing $c>1$, amplification $k$, and approximation $\gamma$
as parameters.

As in Chen and Li~\cite{CL25}, each vertex $v$ becomes
a bidder whose multiplier $\alpha_v$ encodes $\xx[v]$: low
multipliers represent $0$, high multipliers represent $1$, and
intermediate multipliers represent $\bot$. A gate makes the output
bidder's payment respond to its inputs: low input bids leave slack
and force the output multiplier up, while high input bids raise
prices and force it down. In their setting of parallel second-price auctions~\cite{CL25},
it is straightforward to give each gadget dedicated items valued
only by its input and output bidders.
In sponsored shopping, all items enter one
ranking, so an item in one gate may set prices in another.
Moreover, all of a bidder's input and output items share one
budget. We must control both these price interactions and the
aggregate payments. Our construction will make sure that all items are displayed with click-through
rate one, so only prices matter.

Now suppose the gates are ordered as $g_1,\ldots,g_m$. Our goal is to create enough separation between gates to isolate their price interactions. To achieve this, we assign \emph{levels} to their outputs in this order, leaving gaps of at least two between occupied levels.
Items at level
$\ell$ have values $c^{\ell+\theta}$, with fixed offsets
$\theta\in(0,1)$ determining their roles. For now, suppose every
real multiplier lies in $[c^{-1},1]$. Its bids then lie between
$c^{\ell-1}$ and $c^{\ell+1}$, so items from different levels
cannot interleave. We also introduce a \emph{guard} bidder with
items of values $c^1,c^2,\ldots$, covering all levels of the
construction. Its budget is large enough that approximate maximal
pacing forces its multiplier to one, so it bids exactly these values.

To make a gate's effect dominate payments elsewhere, give its
output bidder a block of $k$ identical items for some large number $k$. Under GSP without
self-pricing, these items cannot price one another, so a competing
bid just below the block can price all $k$ items at once. This
amplifies a single input's effect by a factor of $k$, which we
leave unspecified for now.

Consider a $\NOT$ gate with input $u$ and output $v$ at level
$\ell$. Give $v$ a block of $k$ items of value $c^{\ell+2/3}$ and
a budget contribution of $kc^{\ell+1/12}$, and give $u$ an item
of value $c^{\ell+1/3}$ and a budget contribution of $c^{\ell-1}$,
as for every input item at level $\ell$.
Decode $\alpha\leq c^{-1/3}$ as $0$ and
$\alpha=1$ as $1$. If $\xx[u]=0$, the input bid is at most
$c^\ell$, so the output block pays at most $kc^\ell$. For
sufficiently large $k$, this slack dominates the payments on
$v$'s input items, so approximate maximal pacing forces
$\alpha_v=1$. If $\xx[u]=1$, the input bid
is $c^{\ell+1/3}$. If $\alpha_vc^{\ell+2/3}>c^{\ell+1/3}$,
then every output item pays at least $c^{\ell+1/3}$, so the
output block pays at least $kc^{\ell+1/3}>kc^{\ell+1/12}$,
exceeding its budget contribution. Each of $v$'s input items at
a level $j$ bids strictly above the guard bid $c^{j-1}$, so it
pays at least its budget increment $c^{j-1}$. Thus $v$'s total
payment would exceed its budget. Budget feasibility
therefore forces $\alpha_vc^{\ell+2/3}\leq c^{\ell+1/3}$,
or equivalently $\alpha_v\leq c^{-1/3}$, giving $\xx[v]=0$.

It remains to choose the parameters. Let $V_v^{\mathrm{in}}$
be the total value of bidder $v$'s input items and $\ell_v$ its
output level. Low input bids leave the required budget slack if
\[
 k\bigl((1-\gamma)c^{1/12}-1\bigr)
 >
 \frac{V_v^{\mathrm{in}}}{c^{\ell_v}}
 \qquad\text{for every }v.
\]
This also enforces the multiplier lower bound: if
$\alpha_v<c^{-1}$, the output block pays less than $kc^{\ell_v}$,
leaving enough slack to force $\alpha_v=1$.

For an arbitrary gate ordering, one valid choice is
\[
 \gamma=\frac{1}{36m},\qquad
 c=\left(1+\frac{1}{36m}\right)^{36},\qquad
 k=Cm^2
\]
for a sufficiently large absolute constant $C$.
Since there are $O(m)$ levels, $c^{O(m)}=O(1)$ and
$V_v^{\mathrm{in}}/c^{\ell_v}=O(m)$, while
$(1-\gamma)c^{1/12}-1=\Theta(1/m)$. Thus quadratic $k$ suffices,
giving a polynomial-size reduction and hardness at
inverse-polynomial precision.

To obtain hardness for a fixed positive $\gamma$, we need the
slack condition to hold with $k$ polynomial in $m$, since the
reduction explicitly creates $k$ items for each output block.
The slack condition now requires $c>(1-\gamma)^{-12}$, bounded
away from one. Choosing $c$ as a constant depending only on
$\gamma$ makes $(1-\gamma)c^{1/12}-1$ a positive constant, so a
polynomial bound on $V_v^{\mathrm{in}}/c^{\ell_v}$ would suffice.
However, a use of $v$ at level $\ell_v+d$ creates an input item
that contributes $c^{d+\theta}$ to this ratio. Thus a dependency
spanning $\Theta(m)$ levels can make the ratio, and hence the
required number $k$ of output items, exponential in $m$.
The bottleneck is a \emph{long forward dependency}. Uses at lower
levels are harmless because their contributions to the ratio
form a bounded geometric sum.

Surprisingly, we strengthen the hardness result from
inverse-polynomial to constant approximation by establishing
hardness for a more restricted version of \PureCircuit, motivated by the observation of the construction above.
Specifically, we introduce \PureCircuit\ with \emph{bounded forward
span}: every forward wire spans only constantly many gates, while
backward wires are unrestricted. We prove that this variant remains
\PPAD-hard (\Cref{thm:ordered-source}). Assigning levels in this
order bounds every input item's level by $\ell_v+O(1)$.
Since each level contains only constantly many input items,
summing their values gives $V_v^{\mathrm{in}}/c^{\ell_v}=O_\gamma(1)$.
This is stronger than the polynomial bound needed: it lets us
choose $k$ depending only on $\gamma$. The resulting reduction
has polynomial size for every fixed $0\leq\gamma<1$, yielding tight
hardness. The stronger guarantee
comes from tailoring the source circuit to the auction's value scales.

\paragraph{ROI constraints.}
The ROI reduction requires multiplier bounds that do not depend
on the instance-dependent cap $A$. We introduce two guards with
identical geometric value lists, each augmented by an item worth
$a$ times the largest level value, where $a\geq2$ depends only
on $\gamma$. The guards bound each other's multipliers: equal
multipliers are at most two by their combined ROI constraints;
unequal multipliers force the smaller one to equal one, after
which the extra top item bounds the larger by $a$. A guard item
one level above each output block then bounds its bidder's
multiplier: outbidding it would make all $k$ output items pay
so much that their total payment exceeds the bidder's total value.
Bounded forward span controls the input-item values needed for
this comparison. The resulting bounds lie below $A$, so
approximate maximal pacing requires each real bidder to pay at
least $(1-\gamma)$ times its value. With suitable value offsets,
low input bids therefore force high output multipliers, while
high input bids force low output multipliers by ROI feasibility.

\paragraph{The \PureCircuit\ variant.}
The remaining ingredient is the \PPAD-hardness of \BFSPureCircuit,
which we establish in \Cref{sec:bfs-purecircuit}.
Our reduction builds on the reduction in the
\PureCircuit\ paper~\cite{DFHM24b}, which uses a version of
\textsc{StrongSperner} with polynomial grid side length and polynomial
dimension. In contrast, we use a version with exponential grid side
length and only two dimensions. We also follow their overall proof
structure.

We additionally introduce a Gray-code extraction gadget that takes a
single input in $\{0,1,\bot\}$ and produces $b$ outputs in
$\{0,1,\bot\}$, at most one of which is $\bot$. More importantly, all
pure completions of its output decode under the Gray code to at most
two adjacent integers. The gadget also satisfies the boundary
conditions required by the selection lemma in the \PureCircuit\ paper.
Together with the weak sorting technique used in that paper, these
properties establish the correctness of our reduction.

Furthermore, only two feedback chains in our reduction need to be
accounted for when bounding the forward span. To this end, in the
selection stage, we carefully design these feedback chains and the
ordering of their gates to satisfy the ordering requirements of
\BFSPureCircuit. This completes our reduction.

\section{\FIXP-membership}
\label{sec:fixp-membership}

For each bidder $i\in[n]$, partition its items into groups according
to their values. Let $G_i$ be the number of distinct values of
bidder $i$, and index these values in strictly decreasing order:
\[
v_{i,1}>v_{i,2}>\cdots>v_{i,G_i}>0.
\]
The pair $(i,j)$ identifies the group consisting of all items of
bidder $i$ whose value is $v_{i,j}$.

Let
\[
m_{i,j}\in\mathbb{Z}_{>0}
\]
denote the number of items in group $(i,j)$, which we call its
\emph{multiplicity}. For every bidder $i$,
\[
\sum_{j=1}^{G_i} m_{i,j}=m_i.
\]

Define the set of all groups by
\[
\mathscr G
:=
\bigl\{(i,j): i\in[n],\ j\in[G_i]\bigr\},
\]
and let
\[
T:=|\mathscr G|=\sum_{i=1}^{n}G_i
\]
be the total number of groups. 
Without loss of generality, suppose $T \geq 2$.
For a group $g=(i,j)\in\mathcal{G}$, use the shorthand
\[
\iota(g):=i,
\qquad
v_g:=v_{i,j},
\qquad
m_g:=m_{i,j}.
\]
Here, $\iota(g)$ is the bidder owning group $g$, $v_g$ is the common
value of its items, and $m_g$ is its multiplicity. For each group $g\in\mathscr{G}$, 
fix an internal ordering of its $m_g$ items, for example by increasing original item index.

\subsection{Group Allocations}

We construct a randomized allocation by inserting the value
groups one at a time, treating each group as a single block.
Fix an order
\[
g_1,\ldots,g_T
\]
of the groups, for example, the lexicographic order of their
indices $(i,j)$. This insertion order is independent of the
multiplier profile and need not coincide with the final
allocation order. Once all groups have been inserted, expanding
each group into its constituent items gives a complete allocation.

For each group $g_t$, the circuit computes its bid through
a multiplication-by-constant gate whose output is
\[
b_t(\alpha)=v_{g_t}\alpha_{\iota(g_t)}.
\]
For each insertion step $t\geq2$, a sorting network, e.g., the AKS~\cite{AKS83} sorting network, takes
the previously computed bid values
\[
b_1(\alpha),\ldots,b_{t-1}(\alpha)
\]
as inputs and outputs
\[
d_{t,1}(\alpha)\geq\cdots\geq d_{t,t-1}(\alpha),
\]
their nonincreasing rearrangement.
Each comparator is implemented by
\[
(x,y)\longmapsto
(\max\{x,y\},\min\{x,y\}).
\]
Thus, both the group bids and their sorted values are
computed by a continuous piecewise-linear circuit.

At step $t\geq2$, an insertion position
$r\in\{0,\ldots,t-1\}$ means that exactly $r$ previously
inserted groups precede $g_t$.
Define
\[
L_t(\alpha)
:=
\bigl|\{s<t:b_s(\alpha)>b_t(\alpha)\}\bigr|,
\qquad
U_t(\alpha)
:=
\bigl|\{s<t:b_s(\alpha)\geq b_t(\alpha)\}\bigr|.
\]
The legal insertion positions are precisely
\[
\mathcal R_t(\alpha)
:=
\{L_t(\alpha),\ldots,U_t(\alpha)\}.
\]
For notational convenience, set $d_{t,0}(\alpha) := +\infty$ and $d_{t,t}(\alpha) := -\infty$.
Then an position $r$ is legal in step~$t$ iff the
condition hold:
\[
d_{t, r}(\alpha) \geq b_t(\alpha) \geq d_{t,r + 1}(\alpha).
\]
It directly applies to the outputs of sorting network in step~$t$, since
the quantities $L_t$ and $U_t$ are independent of how ties are broken.

Following the pseudogate framework of
\cite{FHHH23,FHHH24},
we may use the Heaviside step function as a gate in circuit constructions.
Each occurrence of this gate is implemented by a
constant-size PL circuit an auxiliary variable,
and these are all included in the state of the full circuit.
Its specified input-output relation is guaranteed
at every fixed point of the full circuit.

For every $t\geq2$ and $j\in[t-1]$, let
$y_{t,j}$ denote the output of a Heaviside
gate with input $b_t(\alpha)-d_{t,j}(\alpha)$.
Thus, at every fixed point,
\[
y_{t,j}\in
\begin{cases}
\{0\},&b_t(\alpha)<d_{t,j}(\alpha),\\
[0,1],&b_t(\alpha)=d_{t,j}(\alpha),\\
\{1\},&b_t(\alpha)>d_{t,j}(\alpha).
\end{cases}
\]

The circuit converts the outputs $y_{t,j}$ into insertion
probabilities by successively assigning the remaining
probability mass.
Starting from a constant node $R_{t,0}=1$, for each
$j=1,\ldots,t-1$, a min gate computes
\[
\rho_{t,j-1}
=
\min\{R_{t,j-1},y_{t,j}\},
\]
and a subtraction gate then computes
\[
R_{t,j}
=
R_{t,j-1}-\rho_{t,j-1}.
\]
The final residual is used directly as the probability
of the last insertion position:
\[
\rho_{t,t-1}=R_{t,t-1}.
\]
Thus, all quantities $R_{t,j}$ and $\rho_{t,r}$ are
computed at nodes of the circuit.

\begin{lemma}
\label{lem:legal-insertion-marginals}
Each vector $\rho_t=(\rho_{t,0},\ldots,\rho_{t,t-1})$ satisfies
\[
\rho_{t,r}\geq0,
\qquad
\sum_{r=0}^{t-1}\rho_{t,r}=1.
\]
Moreover, at the fixed point of the circuit, 
\[
\rho_{t,r}>0
\quad\Longrightarrow\quad
r\in\mathcal R_t(\alpha).
\]
\end{lemma}

\begin{proof}
The recursion gives $0\leq R_{t,j}\leq1$,
$\rho_{t,r}\geq0$, and
\[
\sum_{r=0}^{t-1}\rho_{t,r}
=
\sum_{j=1}^{t-1}(R_{t,j-1}-R_{t,j})
+R_{t,t-1}
=1.
\]
Consider now a fixed point of the circuit such that $\rho_{t,r}>0$.
If $r < t - 1$, since
\[
0<\rho_{t,r}
=\min\{R_{t,r},y_{t,r+1}\}
\]
it follows that $y_{t,r+1}>0$ and $R_{t, r} > 0$.
By the definition of the Heaviside function, $y_{t,r+1} > 0$ implies
\[
b_t(\alpha)\geq d_{t,r+1}(\alpha).
\]
If $r = t- 1$, this follows from $d_{t,t} = - \infty$.
If $r > 0$,
by the definition of the remaining probability mass, $R_{t, r} > 0$ 
implies $\rho_{t, r-1} < R_{t, r-1}$. As $\rho_{t,r-1} = \min \{R_{t,r-1}, y_{t,r}\}$,
we obtain $y_{t,r} < R_{t,r-1} \leq 1$.
Using the definition of the Heaviside function again,
\[
b_t(\alpha) \leq d_{t, r}(\alpha).
\]
If $r = 0$, this follows from $d_{t,0} = + \infty$. Thus, $r \in \mathcal R_t(\alpha)$.
\end{proof}

For each step $t\geq2$, let $e_{t,i}\in \{0,1\} ^t$
denote the unit vector corresponding to insertion position
$i\in\{0,\ldots,t-1\}$.
Every deterministic group allocation $\sigma$ uniquely
determines a sequence of insertion choices as \cite[Definition 2]{LB14}, and hence
an encoding
\[
(e_{2,\sigma(2)},\ldots,e_{T,\sigma(T)}),
\]
where $\sigma(t)$ is the position index in step~$t$. 
A randomized allocation $\pi$ is a probability mixture
of deterministic allocations, and therefore of their
insertion encodings.
Its marginal distribution over insertion positions
at step $t$ is
\[
p_t
=
\sum_{\sigma}\pi(\sigma)e_{t,\sigma(t)}.
\]

We will show that if these marginals satisfy the condition
\[
p_{t,r}>0
\quad\Longrightarrow\quad
r\in\mathcal R_t(\alpha)
\]
for every $2 \leq t \leq T$ and $0 \leq r \leq t-1$, then every
such joint distribution induces a bid-valid randomized
allocation.
Intuitively, an insertion sequence with positive probability
can only use positions having positive marginal probability.
All its insertion choices are therefore legal, and executing
them preserves the nonincreasing bid order at every step.

At every fixed point of the circuit,
\Cref{lem:legal-insertion-marginals} guarantees that
the vectors $\rho_2,\ldots,\rho_T$ are probability
distributions supported on legal insertion positions.
To obtain a randomized allocation, we combine these
marginals into a joint distribution over complete
insertion sequences.
Such a joint distribution, whose marginal at each step
$t$ is $\rho_t$, is called a \emph{coupling} of
$\rho_2,\ldots,\rho_T$.
Here, we use the \emph{quantile coupling} from
\cite{Tho00} and \cite[Proof of Theorem 4.2.3]{Roc24}.
Our construction replies only on the elegant and remarkably simple idea underlying
the coupling method.

\begin{lemma}
Fix a multiplier profile $\alpha$.
Let $\pi$ be any probability distribution over deterministic
group allocations with insertion marginals $p_2,\ldots,p_T$.
If
\[
p_{t,r}>0
\quad\Longrightarrow\quad
r\in\mathcal R_t(\alpha)
\]
for every $2\leq t\leq T$ and $0\leq r\leq t-1$,
then $\pi$ induces a bid-valid randomized allocation.
\end{lemma}

\begin{proof}
Take any deterministic group allocation $\sigma$
with $\pi(\sigma)>0$.
For each $t=2,\ldots,T$, the marginal identity gives
\[
p_{t,\sigma(t)}
=
\sum_{\tau:\,\tau(t)=\sigma(t)}\pi(\tau)
\geq \pi(\sigma)>0.
\]
Hence $\sigma(t)\in\mathcal R_t(\alpha)$ at every step.

We show by induction that executing the insertion sequence
of $\sigma$ preserves the nonincreasing bid order.
The ordering consisting only of $g_1$ is trivially bid-valid.
Suppose that the ordering after step $t-1$ is bid-valid.
Its numerical bid sequence must then be
\[
d_{t,1}(\alpha),\ldots,d_{t,t-1}(\alpha),
\]
regardless of how groups with equal bids are ordered.
Since $\sigma(t)\in\mathcal R_t(\alpha)$,
\[
d_{t,\sigma(t)}(\alpha)
\geq b_t(\alpha)
\geq d_{t,\sigma(t)+1}(\alpha).
\]
Inserting $g_t$ at position $\sigma(t)$ therefore preserves
the nonincreasing bid order, including at the endpoints
under our conventions for $d_{t,0}$ and $d_{t,t}$.
This completes the induction.
\end{proof}

We now describe the circuit computations used for the
quantile coupling.
The circuit has already computed the vectors
$\rho_2,\ldots,\rho_T$, which are supported on legal
insertion positions at every fixed point by
\Cref{lem:legal-insertion-marginals}.

For each $t$, starting from a constant node $F_{t,-1}=0$,
addition gates compute the cumulative probabilities
\[
F_{t,r}=F_{t,r-1}+\rho_{t,r},
\qquad r=0,\ldots,t-1.
\]
Thus,
\[
F_t=(F_{t,-1},F_{t,0},\ldots,F_{t,t-1})
\]
is the cumulative probability vector associated with
$\rho_t$, including both endpoints
$F_{t,-1}=0$ and $F_{t,t-1}=1$.

Set
\(
Q:=\sum_{t=2}^{T}(t+1)-1.
\)
A sorting network takes all entries of
$F_2,\ldots,F_T$ as inputs and outputs their
nondecreasing rearrangement
\[
0=E_0\leq E_1\leq\cdots\leq E_Q=1.
\]
All repeated values are retained.
For every $h\in[Q]$, subtraction gates compute the
interval lengths
\[
w_h=E_h-E_{h-1},
\]
while addition and multiplication-by-constant gates
compute the midpoints
\[
u_h=\frac12(E_{h-1}+E_h).
\]
These computations are performed for all intervals,
including those of zero length.

\begin{lemma}[Quantile Coupling]
\label{lem:quantile-coupling}
At each step $t$, select the first entry of $F_t$
strictly greater than or equal to $u_h$:
\[
\sigma_h(t)
:=
\min\{r\in\{0,\ldots,t-1\}:F_{t,r}\geq u_h\}.
\]
These choices specify a deterministic group allocation
$\sigma_h$ through its insertion encoding.
The probability mixture assigning weight $w_h$ to
$\sigma_h$ for every $h \in Q$ is a randomized
allocation $\pi$ whose insertion marginals are
$\rho_2,\ldots,\rho_T$.
Moreover, $\pi$ is supported on at most
$Q = O(T^2)$ deterministic allocations.
\end{lemma}

\begin{proof}
The weights are nonnegative and satisfy
\[
\sum_{h=1}^{Q}w_h
=
\sum_{h=1}^{Q}(E_h-E_{h-1})
=E_Q-E_0
=1.
\]
Moreover, $u_h\in[0,1]$ and $F_{t,t-1}=1$, so
$\sigma_h(t)$ is well-defined for every $h$ and $t$.
Each insertion sequence specifies a deterministic group
allocation, and the weights therefore yield a randomized
allocation $\pi$.

We verify that its insertion marginals are
$\rho_2,\ldots,\rho_T$.
Fix a step $t$ and a position $r\in\{0,\ldots,t-1\}$.
By definition, when $w_h > 0$, $\sigma_h(t) = r$ if and only if $F_{t, r-1} < u_h \leq F_{t,r}$.
Suppose $F_{t,r}$ appears as $E_u$ and $F_{t, r-1}$ appears as $E_\ell$ after sorting.
Thus, 
\[
\begin{aligned}
\Pr_{\sigma\sim\pi}[\sigma(t)=r]
& =
\sum_{h\in[Q]}w_h \cdot \mathbf{1}[ F_{t, r-1} < u_h \leq F_{t,r}] \\
& =
\sum_{h\in[Q]}w_h \cdot \mathbf{1}[u_h \leq E_u] - \sum_{h\in[Q]}w_h \cdot \mathbf{1}[u_h \leq E_\ell]\\
& = \sum_{h=1}^{u}(E_h-E_{h-1}) - \sum_{h=1}^{\ell}(E_h-E_{h-1}) \\
& = E_u - E_\ell = F_{t, r} - F_{t, r-1} = \rho_{t,r}.
\end{aligned}
\]
Hence every prescribed insertion marginal is preserved.
\end{proof}

We now implement the insertion encodings using Heaviside
gates.
For every $h\in[Q]$, $t=2,\ldots,T$, and
$r=0,\ldots,t-1$, subtraction gates compute
$u_h-F_{t,r-1}$ and $F_{t,r}-u_h$.
A min gate takes these two values as inputs, and its
output is fed into a Heaviside gate.
Denote the resulting output by $z_{h,t,r}\in[0,1]$.
At every fixed point, these outputs satisfy
\[
z_{h,t,r}
\in
H
\left(
\min\{u_h-F_{t,r-1},\,F_{t,r}-u_h\}
\right).
\]

Suppose that the circuit is at a fixed point and $w_h>0$.
The midpoint $u_h$ then differs from every entry of $F_t$.
Thus,
\[
z_{h,t,r}
=
\begin{cases}
1,&F_{t,r-1}<u_h<F_{t,r},\\
0,&\text{otherwise}.
\end{cases}
\]
Exactly one position satisfies these inequalities,
namely $\sigma_h(t)$ from
\Cref{lem:quantile-coupling}.
Thus the circuit outputs the insertion encoding
\[
z_{h,t}
=
(z_{h,t,0},\ldots,z_{h,t,t-1})
=
e_{t,\sigma_h(t)}.
\]
For zero-weight intervals, the outputs $z_{h,t}$
need not belong to $\{0,1\}$, but such intervals do not contribute
to the allocation distribution.
The construction uses $O(QT^2)=O(T^4)$ gates.

Next is to recover the group positions from insertion encodings.
For each $h\in[Q]$ and $w_h > 0$, arithmetic gates can compute the
$\sigma_h(t) = \sum_{r=0}^{t-1} r z_{h,t,r}$.
For each group $g_t$, we track its position through
the subsequent insertions using $T-t$ circuit layers.
The initial node is
\[
\mathrm{Idx}_{h,t}^{(0)}
=
\sigma_h(t).
\]
For each $k=0,\ldots,T-t-1$, the next layer outputs
\[
\mathrm{Idx}_{h,t}^{(k+1)}
=
\mathrm{Idx}_{h,t}^{(k)}
+
H
\left(
\mathrm{Idx}_{h,t}^{(k)}
+\frac12
-\sigma_h(t+k+1)
\right).
\]
This construction uses $O(T^2)$ gates for each $h$.

\begin{lemma}
The final output
$\mathrm{Idx}_{h,t}=\mathrm{Idx}_{h,t}^{(T-t)}$
records the position of $g_t$ in the deterministic allocation $\sigma_h$
at the fixed point and $w_h > 0$.
\end{lemma}

\begin{proof}
Consider a fixed point
and an index $h$ with $w_h>0$.
Immediately after its insertion, $g_t$ occupies
position $\sigma_h(t)$.
Suppose that $\mathrm{Idx}_{h,t}^{(k)}$ correctly
records its position after step $t+k$.
Inserting $g_{t+k+1}$ shifts $g_t$ one position to
the right exactly when
\[
\sigma_h(t+k+1)
\leq
\mathrm{Idx}_{h,t}^{(k)}.
\]
Both quantities are integers, so the input to the
Heaviside gate \emph{cannot} be zero.
Thus, Heaviside gate is exactly the
indicator of this event, and the addition gate
correctly updates the position.
Induction therefore shows that $\mathrm{Idx}_{h,t}$
is the final position of $g_t$ in $\sigma_h$.
\end{proof}

Finally, expand group allocation into item allocation.
For each $h\in[Q]$ and distinct $s,t\in[T]$, a
Heaviside gate produces an output $\beta_{h,s,t}$
satisfying, at every fixed point,
\[
\beta_{h,s,t}
\in H
\left(
    \mathrm{Idx}_{h,t}-\mathrm{Idx}_{h,s}-\frac12
\right).
\]
Using the fixed internal ordering of each group,
arithmetic gates compute the position of its
$\ell$-th item as
\[
\mathrm{ItemIdx}_{h,t,\ell}
=
\sum_{s\in[T], s\neq t}
m_{g_s}\,\beta_{h,s,t}
+\ell.
\]
At every fixed point with $w_h>0$,
$\beta_{h,s,t}$ equals $1$ exactly when $g_s$
precedes $g_t$ in $\sigma_h$.
Thus the sum counts all items in preceding groups,
and adding $\ell$ gives the item's position,
indexed starting from $1$.

\subsection{GSP Payments}

We first compute the next bid belonging to another
bidder using the group positions.
For every $h\in[Q]$ and $t\in[T]$, compute
\[
N_{h,t}
=
\min\left(
\{T\}\cup
\left\{
\beta_{h,t,s}\mathrm{Idx}_{h,s}
+(1-\beta_{h,t,s})T
:
s\in[T],\ \iota(g_s)\neq\iota(g_t)
\right\}
\right).
\]
At every fixed point with $w_h>0$,
$\beta_{h,t,s}=1$ exactly when $g_s$ follows $g_t$.
Thus, a candidate group contributes its position
if it follows $g_t$, and contributes $T$ otherwise.
Since group positions range from $0$ to $T-1$,
$N_{h,t}$ is the position of the first subsequent
group belonging to another bidder, with $T$
indicating that no such group exists.

Here and below, absolute values are implemented
using max gates.
Then compute
\[
J_{h,t}
=
\sum_{s\in[T]}
b_s(\alpha)
H\left(\frac12-\left|\mathrm{Idx}_{h,s}-N_{h,t}\right| \right).
\]
At every fixed point with $w_h>0$, the group
positions are distinct integers, so these selectors
identify precisely the group at position $N_{h,t}$.
If $N_{h,t}=T$, all selectors are zero.
Consequently, $J_{h,t}$ is the next bid belonging
to another bidder or zero if none exists, as required by the GSP rule.
All items in $g_t$ share this per-click price,
since each group forms a contiguous block.

We next retrieve each item's CTR using its item
position. Using the identical idea,
compute
\[
C_{h,t,\ell}
=
\sum_{k=1}^{K}c_k H\left(
\frac12-
\left|\mathrm{ItemIdx}_{h,t,\ell}-k\right|
\right).
\]
At every fixed point with $w_h>0$, item positions
are integers, so Heaviside function outputs $1$ exactly when
the item occupies position $k$.
Hence $C_{h,t,\ell}$ equals the CTR of that
position, and equals zero for positions beyond $K$.

A multiplication gate computes each item's payment:
\[
p_{h,t,\ell}
=
J_{h,t}\,C_{h,t,\ell}.
\]
For every $h\in[Q]$ and $i\in[n]$, addition gates
aggregate these payments over bidder $i$'s items:
\[
P_{h,i}
=
\sum_{\substack{t\in[T]\\\iota(g_t)=i}}
\sum_{\ell=1}^{m_{g_t}}p_{h,t,\ell}.
\]
Finally, multiplication and addition gates compute
\begin{equation}
\label{eq:mutiplication between variables}
P_i
=
\sum_{h=1}^{Q}w_h P_{h,i}.
\end{equation}

All these computations are performed for every
$h\in[Q]$, including indices with $w_h=0$.
At every fixed point, each positive-weight term
$P_{h,i}$ is bidder $i$'s total GSP payment under
$\sigma_h$ and the bids induced by $\alpha$.
Therefore, $P_i$ is exactly bidder $i$'s expected
GSP payment under the randomized allocation $\pi$.

\subsection{Fixed Point in Budget Case}
Set the updated multiplier profile $\tilde \alpha$ as the outputs of the circuit, 
together with all auxilary varibles $\tilde \zeta $ used in Heaviside function and 
the representation of randomized allocation $\tilde \pi = (w_h, \mathrm{ItemIdx}_h)_{h \in [Q]}$. 
For each $i \in [n]$, use
\[
\tilde \alpha_i = \max \left\{0, \min \left\{1,\alpha_i + B_i - P_i \right\} \right\}.
\]

\begin{lemma}
    When the circuit is at the fixed point, namely $F(\alpha, \zeta, \pi) = (\alpha, \zeta, \pi)$,
    $(\alpha, \pi)$ is alwasy an equilibrium for ABE-B.
\end{lemma}

\begin{proof}
    Let $(\alpha,\zeta,\pi)$ be a fixed point of the
full circuit.
By the preceding construction, $\pi$ is bid-valid
under the bids induced by $\alpha$, and $P_i$
is bidder $i$'s expected GSP payment.
Moreover, for every $i\in[n]$, $\alpha_i\in[0,1]$.

We first establish budget feasibility.
Suppose, towards a contradiction, that $P_i>B_i$
for some bidder $i$.
If $\alpha_i>0$, the update strictly decreases
$\alpha_i$, contradicting the fixed point identity.
Hence $\alpha_i=0$, and all items of bidder $i$
have bid zero.
For every $h$ with $w_h>0$ and every group $g_t$
owned by $i$, bid-validity and nonnegativity of bids
therefore imply
\[
0\leq J_{h,t}\leq b_t(\alpha)=0.
\]
Thus every item of bidder $i$ incurs zero payment
in every allocation in the support of $\pi$.
It follows that $P_i=0$, contradicting
$P_i>B_i>0$.
Consequently, $P_i\leq B_i$ for every bidder.

It remains to verify maximal pacing.
If $P_i<B_i$ and $\alpha_i<1$, then
$\tilde \alpha_i > \alpha_i$,
again contradicting the fixed point identity.
Therefore, $P_i<B_i$ implies $\alpha_i=1$.
These conclude that $(\alpha,\pi)$ is an ABE-B.
\end{proof}

\subsection{Fixed Point in ROI case}
Additionally use $C_{h,t,\ell}$ to compute
bidder $i$'s total valuation $V_{h,i}$ under
allocation $\sigma_h$ and its expected valuation
$V_i$ under $\pi$.
These computations follow the same circuit
construction as the payments, with $v_{g_t}$
replacing $J_{h,t}$.
For each $i\in[n]$, the circuit computes the
updated multiplier as
\[
\tilde \alpha_i = \max \left\{1, \min \left\{A,\alpha_i + V_i - P_i \right\} \right\}.
\]

\begin{lemma}
    When the circuit is at the fixed point, namely $F(\alpha, \zeta, \pi) = (\alpha, \zeta, \pi)$,
    $(\alpha, \pi)$ is always an equilibrium for ABE-R.
\end{lemma}

\begin{proof}
Let $(\alpha,\zeta,\pi)$ be a fixed point of the
full circuit.
By the preceding construction, $\pi$ is bid-valid
under the bids induced by $\alpha$, and $P_i$
and $V_i$ are bidder $i$'s expected payment and
valuation, respectively.
Moreover, $\alpha_i\in[1,A]$ for every $i\in[n]$.

We first establish ROI feasibility.
Suppose, towards a contradiction, that $P_i>V_i$
for some bidder $i$.
If $\alpha_i>1$, the update strictly decreases
$\alpha_i$, contradicting the fixed point identity.
Hence $\alpha_i=1$.
For every $h$ with $w_h>0$ and every group $g_t$
owned by $i$, bid-validity implies
\[
0\leq J_{h,t}\leq b_t(\alpha)=v_{g_t}.
\]
Since $C_{h,t,\ell}\geq0$, each item of $g_t$
therefore satisfies
\[
p_{h,t,\ell}
=
J_{h,t}C_{h,t,\ell}
\leq
v_{g_t}C_{h,t,\ell}.
\]
Summing over bidder $i$'s items gives
$P_{h,i}\leq V_{h,i}$ for every $h$ with $w_h>0$.
Taking the probability-weighted sum yields
$P_i\leq V_i$, contradicting our assumption.
Thus every bidder satisfies the ROI constraint.

It remains to verify maximal pacing.
If $V_i>P_i$ and $\alpha_i<A$, then the update
gives $\tilde{\alpha}_i>\alpha_i$, again
contradicting the fixed point identity.
Therefore, $V_i>P_i$ implies $\alpha_i=A$.
\end{proof}

Additionally, we want to point out a property of both ABE-B and ABE-R which
we believe is helpful for developing concrete algorithms for equlibrium computation
in sponsored shopping. It follows directly from our circuit for each allocation is obtained by first constructing a
permutation of the groups and then expanding each group into its constituent items in a fixed internal order.

\begin{corollary}
In both the budget-constrained and ROI-constrained settings,
there exists an exact equilibrium
in which all items belonging to the same bidder and having the
same value occupy consecutive positions in every deterministic
allocation in its support.
\end{corollary}

\section{\PPAD-membership}
\label{sec:ppad-membership}
A natural question is whether our \FIXP\ upper bound for computing exact autobidding equilibria can be
improved to \PPAD. In this section, we construct rational instances in both the budget-constrained and ROI-constrained
settings for which every exact equilibrium has an irrational pacing multiplier.
Thus, exact equilibria need not admit rational representations, presenting an obstruction to \PPAD-membership. This 
motivates considering approximate equilibria, for which we establish \PPAD-membership in the following subsections.

\subsection{Instances with Only Irrational Equilibria}

\begin{theorem}
    There exists a GSP shopping auction with rational bidder values, rational finite budgets, and rational 
    click-through rates such that every autobidding equilibrium under budget constraints has an irrational pacing 
    multiplier. The corresponding problem of computing an exact autobidding equilibrium is therefore not in \PPAD\
    under standard rational representation.
\end{theorem}

\begin{proof}
Suppose there are three agents and six items for sale with all click-through rates equal to $1$. The pacing
multiplier of each agent is denoted as $\alpha_i$. The budget of
each agent is denoted as $B_i$.
\begin{itemize}
\item Agent~1 has items $1000, 1$ and budget $B_1 = 100$.
\item Agent~2 has items $2000, 300, 10$ and budget $B_2 = 30$.
\item Agent~3 has items $10$ and an unbounded budget $B_3 = 1000$.
\end{itemize}

We discuss possible equilibria case by case:
\begin{itemize}
\item[(a)] Suppose, for contradiction, the payment of agent~1 is strictly less than $100$ and the payment of agent~2 strictly less than $30$. Each agent should be maximally paced, thus $\alpha_1 = \alpha_2 = 1$. All bids
can be sorted as $2000, 1000, 300, 10, 10, 1$. Obviously, payments of both
agent~1 and agent~2 exceed their budgets.
\item[(b)] Suppose the payment of agent~1 strictly less than $100$ and the payment of agent~2 equals $30$ for contradiction. Agent~1 should be maximally paced, thus $\alpha_1 = 1$. Agent~2 must spend under the budget, thus the bid $1000 \alpha_1 = 1000$ from agent~1 must be higher than any bid from agent~2. Additionally, the bid $10$ from agent~3 can't be higher than any bid from agent~2, otherwise the payment of agent~2 will be strictly less than $30$. Thus $\alpha_2 = 1$, which violates  $2000\alpha_2 > 1000\alpha_1$.
\item[(c)] Suppose that the payment of agent~1 equals to $100$ and the payment of agent~2 strictly less than $30$. Agent~2 must be maximally paced, thus $\alpha_2 = 1$. With the restriction of payment of agent~1, the only possible bid ranking is $2000, 1000\alpha_1=300, 10, 10, \alpha_1$. But it leads to the payment of agent~2 being greater than $30$. 
\item[(d)] Thus, the only possible equilibrium can be attained when the payment of agent~1 equals to $100$ and the payment of agent~2 equals to $30$. Compare the highest bid of Agent~1 and Agent~2:
  \begin{itemize}
  \item[(i)] When $2000\alpha_2 > 1000\alpha_1$, agent~1 wins the last slot for free and spends the whole budget on the bid $1000\alpha_1$. By the definition of second-price auction, $1000\alpha_1 \geq 100$ and then $2000\alpha_2 > 100$. Thus the payment of agent~2 would exceed $30$.
  \item[(ii)] When $2000\alpha_2 < 1000\alpha_1$, the ranking of bids from agent~1 and agent~3 is fixed as $1000\alpha_1, 10, \alpha_1$. No bid from agent~2 can lie between $10$ and $\alpha_1$, otherwise the payment of agent~2 would be strictly less than $30$. As $\alpha_2 \in [0,1]$, it implies $\alpha_2 = 1$ and a contradiction $2000\alpha_2 < 1000\alpha_1$.
  \item[(iii)] Finally, the only possible case is $2000\alpha_2 = 1000 \alpha_1$. Then the ranking of bids is now $2000\alpha_2 = 1000\alpha_1, 300\alpha_2, 10, \allowbreak 10\alpha_2, \alpha_1$ as $1000\alpha_1 \geq 100$. The randomized allocations are as follows:
  \begin{itemize}
  \item[(1)] $1,2,2,3,2,1$ with second-price $1000\alpha_1, 10, 10,
    5\alpha_1, \alpha_1, 0$.
  \item[(2)] $2,1,2,3,2,1$ with second-price $1000\alpha_1, 150\alpha_1,
    10, 5\alpha_1, \alpha_1, 0$.
  \end{itemize}
  Denote the probability of case~$(1)$ as $x$ and the probability of
  case~$(2)$ as $1 - x$. Calculate the payments of agents
  individually to get two equations. For agent~1, the payment is
  \[
  1000\alpha_1x + 150\alpha_1(1 - x) = 100.
  \]
  For agent~2, the payment is
  \[
  (20 + \alpha_1)x + (10 + 1000\alpha_1 + \alpha_1)(1 - x) = 30.
  \]
  Thus,
  $$
  170x^2 - 2310x + 1942 = 0  
  $$
  has a single valid root
  $$
  x = 231/34 - \sqrt{200777/5}/34.
  $$
\end{itemize}
\end{itemize}
The corresponding parameter values are approximately $\alpha_1 \approx 0.109254$ and $x \approx 0.900349$. 
\end{proof}

\begin{theorem}
    There is a GSP shopping auction with rational bidder values and rational click-through rates such that every 
    autobidding equilibrium under ROI constraints has an irrational pacing multiplier. The corresponding problem of 
    computing an exact autobidding equilibrium is therefore not in \PPAD\
    under standard rational representation.
\end{theorem}

\begin{proof}
Set all click rates are $1$ and $A \geq 5$. The game is constructed as follows. The pacing
multiplier of each agent is denoted by $\alpha_i$.
\begin{itemize}
\item Agent~1 has items $1000, 1$ with multiplicity $1$.
\item Agent~2 has items $2000, 10$ with multiplicity $1$.
\item Group~3 has $4999$ agents and each agent has item $4998$.
\item Group~4 has $2$ agents and each agent has an item $5$.
\end{itemize}

We first show that all agents in group~3 must have the same multiplier. 

Let \(v=4998\), and let \(\bar\alpha\) be the largest multiplier among agents in group 3. 
Suppose, for contradiction, that there is another group-3 bid \(b'\) satisfying $v < b' < \bar\alpha v$.
Consider any group-3 agent whose bid is \(\bar\alpha v\). If this agent is not ranked last among the agents tied at bid \(\bar\alpha v\), then its GSP price is \(\bar\alpha v>v\). If it is ranked last among this tied class, then its GSP price is at least \(b'>v\). Hence its expected payment is strictly larger than \(v\), violating ROI feasibility. Therefore every group-3 bid below the maximum group-3 bid is at most \(v\). Since multipliers are at least \(1\), every such bid must be exactly \(v\), i.e., the corresponding multiplier is \(1\). Thus, if the group-3 multipliers are not all identical, they can only take the two values \(1\) and \(\alpha_3\). 
For any group-3 agent with multiplier \(1<A\), ROI cannot be slack, so its expected payment must equal \(v=4998\). This is possible only if some non-group-3 bid is also equal to \(4998\); otherwise at least one such group-3 agent is ranked last among the bids equal to \(4998\) with positive probability and pays strictly less than \(4998\), contradicting maximal pacing.

Moreover, only the bids on the highest item of agent~1 and the highest item of agent~2 can equal $4998$, and $1000\alpha_1 \leq 2000\alpha_2 \leq 4998$ always holds by the ROI feasibility of agent~1 and agent~2. If $1000\alpha_1 < 2000\alpha_2 = 4998$, then the maximal pacing rule for agent~1 is violated. Thus $1000\alpha_1 = 2000\alpha_2 = 4998$ in this case. Let $y_1$ denote the probability that the bid $1000\alpha_1$ is ranked lowest among the bids equal to $4998$, and let $y_2$ denote the corresponding probability for the bid $2000\alpha_2$. Notice that $y_1 + y_2 = 1$. However, the total payment of agent~1 and agent~2 is then at least $4998(1 - y_1) + 4998(1 - y_2) = 4998$, which exceeds their total valuation of $3011$. Thus all agents in group~3 share the same multiplier $\alpha_3$.

In our example, one can compute a bound on $\alpha_3$. Since there is a tie, let $x_i$ denote the probability that the bid of agent~$i$ in group~3 is ranked lowest under the randomized tie-breaking rule, and let $p$ denote the largest bid that is less than $4998\alpha_3$. Notice that $\sum_i x_i \leq 1$, since other agents may share the same valuation. By the ROI feasibility of group~3, we have 
\[
x_ip + (1 - x_i)\alpha v \leq v. 
\]
Thus,
\[
\sum_{i=1}^{n} (x_ip + (1 - x_i)\alpha v) \leq nv.
\]
Moreover, if $\alpha \neq A$, we obtain the equation
\begin{equation}
\label{eq:expression-of-multiplier}
\sum_{i=1}^{n} (x_ip + (1 - x_i)\alpha v) = nv.
\end{equation}
Finally,
\[
\alpha_3 \leq \frac{4998 \cdot 4999 - p\sum_{i=1}^{4999}x_i}{4998 \cdot 4999 - 4998\sum_{i=1}^{4999} x_i} \leq \frac{4999}{4998}. 
\]
Similarly, $1000\alpha_1 \leq 2000\alpha_2 \leq 4998\alpha_3$ always holds by the ROI feasibility of agent~1 and agent~2. We must have $\alpha_1 = 2\alpha_2$; otherwise, since $\alpha_2 \leq \frac{4999}{2000}$, the payment of agent~1 would be much smaller than its valuation, violating the maximal pacing rule for agent~1. If $1000\alpha_1 = 2000\alpha_2 = 4998\alpha_3$, then the total payment of agent~1 and agent~2 is at least $4998\alpha_3(1 - y_1) + 4998\alpha_3(1 - y_2) > 4998$, which again exceeds their total valuation of $301
1$. Thus we have $1000\alpha_1 = 2000\alpha_2 < 4998\alpha_3$.

Since $\alpha_1 < 5$, the agents in group~4 must share the same multiplier $\alpha_4$ by an identical argument. Clearly, $\alpha_4$ is bounded and strictly less than $A$.
One can obtain the expression for $\alpha_4$ using \Cref{eq:expression-of-multiplier}. 
\[
\alpha_4 = (2v_4 - \alpha_1)/v_4 = 2 - \alpha_1/5.
\]
The only relevant orderings in our example are as follows:
\begin{itemize}
    \item $4998\alpha_3 > 1000\alpha_1 = 2000\alpha_2 > 10\alpha_2 > 5\alpha_4 > \alpha_1$ with probability $x$,
    \item $4998\alpha_3 > 2000\alpha_2 = 1000\alpha_1 >  10\alpha_2 >  5\alpha_4 > \alpha_1$ with probability $1-x$.
\end{itemize}
By the ROI rule for agent~1, 
\[
1000\alpha_1x + 10\alpha_2(1-x) = 1000\alpha_1 x + 5\alpha_1(1-x) = 1001 
\]
and for agent~2,
\[
10\alpha_4x + 1000\alpha_1(1-x) + 5\alpha_4(1-x) = 1000\alpha_1(1-x) + 5(2 - \alpha_1/5)(1+x)= 2010.
\]
Thus, we obtain the quadratic equation
\[
999010 \alpha_1^2 - 2992051 \alpha_1 + 10010 = 0
\] whose only valid root is $\alpha_1 = \frac{2992051 + \sqrt{8912368826201}}{1998020}$. The corresponding parameter values are approximately $\alpha_1 \approx 2.99167$ and $x \approx 0.331252$.

\end{proof} 

We next identify the operations in our $\FIXP$-membership
construction that require approximation.
All computations preceding the GSP payment calculation already
admit PL pseudo-circuit implementations: they use only addition,
subtraction, multiplication by rational constants, min/max
gates, and Heaviside pseudogates.

Moreover, by \cite[Lemma~3.1]{FHHH24}, the correspondence
$(x,y)\mapsto y \cdot H(x)$ can be implemented by a PL pseudo-circuit
for the purpose of proving $\PPAD$-membership.
This allows us to implement the products appearing
in the computations of $N_{h,t}$ and $J_{h,t}$.
For the item payments, we distribute the product
$J_{h,t}C_{h,t,\ell}$ over the two sums to obtain
\[
p_{h,t,\ell}
=
\sum_{s\in[T]}\sum_{k=1}^{K}
c_k b_s(\alpha)
H\left(
\frac12-\left|\mathrm{Idx}_{h,s}-N_{h,t}\right|
\right)\cdot
H\left(
\frac12-\left|\mathrm{ItemIdx}_{h,t,\ell}-k\right|
\right).
\]
Each summand can be computed by two successive applications
of the above pseudo-circuit, followed by multiplication by
the constant $c_k$. Addition gates then compute $p_{h,t,\ell}$
and $P_{h,i}$.

Consequently, the remaining products requiring approximation
are those between the probability weights and the
allocation-dependent payments in \Cref{eq:mutiplication between variables}, that is 
\[
P_i=\sum_{h=1}^{Q}w_hP_{h,i}.
\]

\subsection{Approximate Equilibria in Budget case}

Fix $0<\gamma<1$. For each bidder $i$, let
\[
U_i\coloneq \sum_{\substack{t\in[T]\\\iota(g_t)=i}}
m_{g_t}v_{g_t}.
\]
At every fixed point, bid-validity and $\alpha_i\leq1$
imply $0\leq P_{h,i}\leq U_i$ whenever $w_h>0$.
Choose the smallest power of two $\aleph =2^d$ satisfying
\[
\aleph \geq
\max\left\{1,\max_{i\in[n]}\frac{U_i}{\gamma B_i}\right\}.
\]

Let $C_{\mathrm{id}}$ be the identity Boolean circuit on
$d+1$ bits, representing the integer-valued function
$f(j)=j$ for $j\in\{0,\ldots,\aleph \}$.
Following \cite[Definition~3.6]{FHHH24}, let
$\mathcal D_\aleph$ denote the piecewise-constant correspondence
represented by $(C_{\mathrm{id}},\aleph )$.
In particular, 
\[
        \mathcal D_\aleph(x) = \begin{cases}
        0 & \text{for } x \leq 0 \\
        [x-1,x] & \text{for } x \in \{1,2, \dots, \aleph \}\\
        
        \lfloor x \rfloor & \text{for } x \in [0, \aleph ]\setminus \{1,2,\dots, \aleph \} \\
        \aleph & \text{for } \aleph < x 
        \end{cases}.        
\]

For each $h\in[Q]$ and $i\in[n]$, apply
\cite[Proposition~3.3]{FHHH24} directly. Specifcally, the multiplication of variables
$w_h \cdot P_{h,i}$ in \Cref{eq:mutiplication between variables} can be replaced a PL pseudo-circuit, namely,
\[
\widehat q_{h,i}
\in
w_h \cdot 
\mathcal D_\aleph\left(\frac{P_{h,i}}{U_i/\aleph}\right) \cdot (U_i/ \aleph) .
\]
Here, $U_i / \aleph$ is a constant.
When $w_h>0$, the approximation satisfies
\[
0\leq w_hP_{h,i}-\widehat q_{h,i}
\leq\frac{U_iw_h}{\aleph}.
\]
When $w_h=0$, the output $\widehat q_{h,i}$ is zero,
so the same inequality holds.
Addition gates compute
\[
\widehat P_i=\sum_{h=1}^{Q}\widehat q_{h,i}.
\]
Since $\sum_h w_h=1$, we obtain
\[
0\leq P_i-\widehat P_i
\leq\frac{U_i}{\aleph}
\leq\gamma B_i.
\]

Set the updated multiplier profile $\tilde \alpha$ as the outputs of the circuit, 
together with all auxilary varibles $\tilde \zeta $ used in Heaviside function and 
the representation of randomized allocation $\tilde \pi = (w_h, \mathrm{ItemIdx}_h)_{h \in [Q]}$. 
For each $i \in [n]$, use
\[
\tilde \alpha_i = \max \left\{ 0, \min \left\{1,\alpha_i + (1-\gamma)B_i - \widehat P_i \right\} \right\}.
\]

\begin{lemma}
    When $0 < \gamma < 1$ has polynomial bit sizes, the circuit can be constructed in polynomial time. Moreover,
    when the circuit is at the fixed point, $(\alpha, \pi)$ is $\gamma$-approximate equilibrium for ABE-B.
\end{lemma}
In particular, the construction works with constant,
inverse-polynomial, and inverse-exponential approximation
parameters.

\begin{proof}
Since $\gamma$ has polynomial encoding length,
the quantities $U_i/(\gamma B_i)$ also have polynomial
encoding length.
Thus, $C_{\mathrm{id}}$ has polynomial size, and all rational
constants used in the construction have polynomial encoding
length. By \cite[Proposition~3.3]{FHHH24}, it admits a PL pseudo-circuit
implementation constructible in polynomial time.

Now fix a fixed point $(\alpha,\zeta,\pi)$ of the circuit.
The allocation construction guarantees that $\pi$ is
bid-valid under the bids induced by $\alpha$, with
$\alpha\in[0,1]^n$.
Moreover, the preceding approximation bound gives
\[
\widehat P_i \geq P_i-\gamma B_i.
\]

We first verify budget feasibility.
Suppose, towards a contradiction, that $P_i>B_i$
for some bidder $i$. Then
\[
\widehat P_i
\geq P_i-\gamma B_i
>(1-\gamma)B_i.
\]
If $\alpha_i>0$, the multiplier update strictly decreases
$\alpha_i$, contradicting that $\alpha$ is a fixed point.
Then, $\alpha_i=0$ must hold.
All items of bidder $i$ then have bid zero.
By bid-validity and nonnegativity of bids, every such item
has GSP price zero in every allocation in the support of
$\pi$. Hence $P_i=0$, contradicting $P_i>B_i>0$.
Therefore, $P_i\leq B_i$ for every bidder.

Finally, suppose that
$P_i<(1-\gamma)B_i$.
Since $\widehat P_i\leq P_i$, we also have
$\widehat P_i<(1-\gamma)B_i$.
If $\alpha_i<1$, the multiplier update strictly increases
$\alpha_i$, again contradicting the fixed point identity.
Thus $\alpha_i=1$, establishing approximate maximal pacing.
Together with bid-validity and budget feasibility, this
shows that $(\alpha,\pi)$ is a $\gamma$-approximate ABE-B.
\end{proof}

\subsection{Approximate Equilibria in ROI case}

Fix $0<\gamma<1$. In the ROI setting, we require the
payment approximation error to be bounded relative to
each bidder's expected valuation. To achieve this, we
first compute the probability-weighted CTRs and then
apply the same discretization correspondence to the
per-click prices.

For every $h\in[Q]$ and $t\in[T]$, the circuit computes
\[
S_{h,t}
=
\sum_{\ell=1}^{m_{g_t}}\sum_{k=1}^{K}
c_k\left[
w_h H\left(
\frac12-
\left|\mathrm{ItemIdx}_{h,t,\ell}-k\right|
\right)
\right].
\]
Each product inside the brackets is implemented using
\cite[Lemma~3.1]{FHHH24}, so this computation
admits a PL pseudo-circuit implementation.
At every fixed point with $w_h>0$, these quantities satisfy
\[
S_{h,t}
=
w_h\sum_{\ell=1}^{m_{g_t}}C_{h,t,\ell}.
\]
Rewrite the true expected
payments and valuations as
\[
P_i
=
\sum_{h=1}^{Q}
\sum_{\substack{t\in[T]\\\iota(g_t)=i}}
J_{h,t}S_{h,t},
\qquad
V_i
=
\sum_{h=1}^{Q}
\sum_{\substack{t\in[T]\\\iota(g_t)=i}}
v_{g_t}S_{h,t}.
\]
Since each $v_{g_t}$ is a constant, the circuit computes
$V_i$ exactly using addition and multiplication-by-constant
gates.

Choose the smallest power of two $\aleph=2^d$ satisfying
\[
\aleph\geq\frac{A}{\gamma},
\]
and use the correspondence $\mathcal D_\aleph$
from the preceding subsection.
For each $h\in[Q]$ and $t\in[T]$, apply
\cite[Proposition~3.3]{FHHH24} to implement
the scaled correspondence whose output satisfies
\[
\widehat q_{h,t}
\in
\frac{Av_{g_t}}{\aleph}\,S_{h,t}\,
\mathcal D_\aleph\left(
\frac{J_{h,t}}{Av_{g_t}/\aleph}
\right)
\]
at every fixed point.
Here, $Av_{g_t}/\aleph$ is a positive rational constant.

The discretization gives
\[
0\leq J_{h,t}S_{h,t}-\widehat q_{h,t}
\leq
\frac{A}{\aleph}v_{g_t}S_{h,t}.
\]
If $S_{h,t}=0$, then $\widehat q_{h,t}=0$,
so the same inequality holds.
Addition gates compute
\[
\widehat P_i
=
\sum_{h=1}^{Q}
\sum_{\substack{t\in[T]\\\iota(g_t)=i}}
\widehat q_{h,t}.
\]
Summing the error bounds yields
\[
0\leq P_i-\widehat P_i
\leq
\frac{A}{\aleph}
\sum_{h=1}^{Q}
\sum_{\substack{t\in[T]\\\iota(g_t)=i}}
v_{g_t}S_{h,t}
=
\frac{A}{\aleph}V_i
\leq\gamma V_i.
\]
This bound also holds when $V_i=0$, in which case
$P_i=\widehat P_i=0$.

As in the budget case, the circuit outputs the
allocation representation, the auxiliary coordinates
of all pseudogates, and the updated multipliers
\[
\tilde\alpha_i
=
\max\left\{1,\min\left\{
A,\alpha_i+(1-\gamma)V_i-\widehat P_i
\right\}\right\}.
\]

\begin{lemma}
For every rational $0<\gamma<1$ whose encoding length
is polynomial in the instance size, the above circuit
can be constructed in polynomial time.
Every fixed point of this circuit induces a
$\gamma$-approximate ABE-R $(\alpha,\pi)$.
\end{lemma}

\begin{proof}
The efficiency follows directly from the cap $A$ specified in the preliminaries has polynomial size. One can show that $d = \lceil \log_2(A / \gamma)\rceil $ has polynomial size.

Now fix a fixed point $(\alpha,\zeta,\pi)$.
We first verify ROI feasibility.
Suppose, towards a contradiction, that $P_i>V_i$
for some bidder $i$. Then
\[
\widehat P_i
\geq P_i-\gamma V_i
>(1-\gamma)V_i.
\]
If $\alpha_i>1$, the multiplier update strictly
decreases $\alpha_i$, contradicting the fixed point
identity. Hence $\alpha_i=1$.
For every $h$ with $w_h>0$ and every group $g_t$
owned by $i$, bid-validity now implies
\[
J_{h,t}\leq b_t(\alpha)=v_{g_t}.
\]
Since $S_{h,t}\geq0$ and $S_{h,t}=0$ whenever $w_h=0$,
we obtain
\[
P_i
=
\sum_{h=1}^{Q}
\sum_{\substack{t\in[T]\\\iota(g_t)=i}}
J_{h,t}S_{h,t}
\leq
\sum_{h=1}^{Q}
\sum_{\substack{t\in[T]\\\iota(g_t)=i}}
v_{g_t}S_{h,t}
=
V_i,
\]
contradicting $P_i>V_i$.
Therefore, every bidder satisfies ROI feasibility.

Finally, suppose that $P_i<(1-\gamma)V_i$.
Since $\widehat P_i\leq P_i$, we also have
$\widehat P_i<(1-\gamma)V_i$.
If $\alpha_i<A$, the multiplier update strictly
increases $\alpha_i$, again contradicting the
fixed point identity.
Thus $\alpha_i=A$, establishing approximate maximal
pacing.
Together with bid-validity and ROI feasibility,
this shows that $(\alpha,\pi)$ is a
$\gamma$-approximate ABE-R.
\end{proof}

\section{\PPAD-hardness of \BFSPureCircuit}
\label{sec:bfs-purecircuit}
This section proves a variant of \PureCircuit\ is \PPAD-hard.
We call this variant \BFSPureCircuit. Our main result establishes
\PPAD-hardness even when the forward span is at most five.
In order to establish the \PPAD-hardness, we present a reduction from
the \TwoDStrongSperner\ problem. 

\subsection{The \TwoDStrongSperner\ problem}

\begin{definition}[\TwoDStrongSperner]
    \label{def:2d-strong-sperner}
    An instance consists of a Boolean circuit $C$ with $2b$ input bits
    and two output bits, where $b\geq 1$, representing a labeling
    \[
    C:\{0,\ldots,2^b-1\}^2\longrightarrow\{0,1\}^2.
    \]
    The labeling satisfies the following boundary conditions:
    for every point $p$ and each coordinate $i\in\{1,2\}$,
    \[
    p_i=0 \Longrightarrow C_i(p)=1,
    \qquad
    p_i=2^b-1 \Longrightarrow C_i(p)=0.
    \]
    The task is to find two points $p,q\in\{0,\ldots,2^b-1\}^2$
    such that
    \[
    \|p-q\|_\infty\leq 1
    \qquad\text{and}\qquad
    C_i(p) \neq C_i(q)
    \quad\text{for each }i\in\{1,2\}.
    \]
\end{definition}

\begin{theorem}
    The \TwoDStrongSperner\ is \PPAD-hard.
\end{theorem}

\begin{proof}
    The reduction from \textsc{End-of-Line}
    to \TwoDStrongSperner\ with exponentially large side length is stated in
    Remark~3 of~\cite{DFHM24b}, using the construction of
    Chen and Deng~\cite{CD09}.
\end{proof}

\subsection{The \BFSPureCircuit\ problem}

The \PureCircuit\ problem, introduced by Deligkas, Fearnley, Hollender, and Melissourgos~\cite{DFHM24b}, will serve as the main intermediate problem for proving the PPAD-hardness results.

\begin{definition}[\PureCircuit]
\label{def:pure-circuit}
An instance of the \PureCircuit\ problem is defined over a vertex set $V$ and a gate set $G$, where $|V|=n$ and $|G|=m\leq n$. The goal is to find an assignment $\mathbf{x}: V\rightarrow \{0,1,\bot\}$ that satisfies the constraints provided by the gates. In particular, each gate $(T,u,v,w)\in G$, where $u,v,w\in V$ are potential inputs and outputs, specifies a constraint defined as follows: 
\begin{itemize}
    \item If $T=\NOT(u;w)$, then $u$ is viewed as the input and $w$ is viewed as the output, and $\mathbf{x}$ satisfies \[\mathbf{x}[u]\in\{0,1\}\ \Longrightarrow \ \mathbf{x}[w]=\neg\mathbf{x}[u].\]
    \item If $T=\NOR(u,v;w)$, then $u$ and $v$ are viewed as inputs and $w$ is viewed as the output, and $\mathbf{x}$ satisfies 
\begin{align*}
\mathbf{x}[u]=1\vee\mathbf{x}[v]=1\ &\Longrightarrow\ \mathbf{x}[w]=0 \quad \text{ and }\\
 \mathbf{x}[u]=0\wedge\mathbf{x}[v]=0\ &\Longrightarrow\ \mathbf{x}[w]=1.
 \end{align*}
	\item If $T=\NPURIFY(u;v,w)$, then $u$ is viewed as the input and $v\neq w$ are viewed as the outputs, and $\mathbf{x}$ satisfies \begin{align*}
	\{\mathbf{x}[v],\mathbf{x}[w]\}\cap\{0,1\}\neq \emptyset \quad \text{ and } \\ \mathbf{x}[u]\in\{0,1\}\ \Longrightarrow\  \mathbf{x}[v]=\mathbf{x}[w]=\neg \mathbf{x}[u].
	\end{align*}
\end{itemize}
Importantly, each node is required to be an output of exactly one gate.
\end{definition}

We also use several auxiliary gates defined in~\cite{DFHM24b}.
Each can be implemented by a constant-size circuit using only the
gates listed in \Cref{def:pure-circuit}. We therefore expand these
auxiliary gates before constructing the auction, so no separate
simulation of them is needed in the auction construction.
The auxiliary gates are implemented as follows. 
\begin{itemize}
\item $\PURIFY(u;v,w)$ is implemented by
\[
 \NPURIFY(u;a,b),\qquad \NOT(a;v),\qquad \NOT(b;w).
\]
At least one of $v,w$ is pure, and a pure input forces both
outputs to equal the input.

\item $\COPY(u;v)$ is implemented by
\[
 \NOT(u;a),\qquad \NOT(a;v).
\]
A pure input is preserved, while an input of $\bot$ permits
any output in $\{0,1,\bot\}$.

\item $\OR(u,v;w)$ is implemented by
\[
 \NOR(u,v;a),\qquad \NOT(a;w).
\]
If either input is one, the output is forced to one.
If both inputs are zero, the output is forced to zero.

\item $\AND(u,v;w)$ is implemented by
\[
 \NOT(u;a),\qquad \NOT(v;b),\qquad \NOR(a,b;w).
\]
If either input is zero, the output is forced to zero.
If both inputs are one, the output is forced to one.
\end{itemize}
In each case, outputs not determined by these conditions may
take any value in $\{0,1,\bot\}$.

\begin{definition}[Forward span]
For an ordering $g_1,\ldots,g_m$ of the gates, a wire from an output
of $g_I$ to an input of $g_J$ has \emph{forward span} $J-I$ when
$J>I$. The ordering has forward span at most $d$ if
\begin{equation*}
 J-I\leq d
\end{equation*}
for every wire. Backward spans are unrestricted.
\end{definition}

\begin{definition}[\BFSPureCircuit]
\label{def:bfs-pure-circuit}
For a fixed positive integer $d$, an instance of \BFSPureCircuit$(d)$
consists of a \PureCircuit\ instance as in \Cref{def:pure-circuit},
together with an ordering of its gates of forward span at most $d$.
Backward spans are unrestricted.
The task is to find a satisfying assignment of the circuit.
\end{definition}

Revisiting the proof of \PPAD-hardness of \PureCircuit\
in~\cite{DFHM24b}, a central step is to simulate the
Boolean labeling circuit of a \textsc{StrongSperner} instance
using \PureCircuit\ gates. Such a simulation is guaranteed to
evaluate the labeling correctly only when its inputs are pure.
To handle impure inputs, the authors introduce a purification
stage that produces sufficiently many fully pure samples.
Their sorting and selection stages, together with feedback,
ensure that each original coordinate block contains at most
one impure bit. Since coordinates are represented in unary,
the pure completions of such a block differ in at most one bit
and therefore represent integers differing by at most one.

This unary representation keeps the reduction polynomial in
size because the grid side length is polynomially bounded.
In our two-dimensional setting, however, the grid side length
is $2^b$, so a unary representation would require exponentially
many input bits. We instead use $b$-bit binary-reflected Gray
codes, together with an extraction gadget that ensures that
each coordinate has at most one impure bit and that its two
pure completions decode to adjacent integers.

For a pure Gray bitstring $\mathbf{g} = (g_1, \ldots, g_b)$,
define its decoded integer by
\[
D_b(\mathbf g)
:=\sum_{j=1}^{b}2^{b-j}
(g_1\oplus\cdots\oplus g_j).
\]
A pure completion of a string in $\{0,1,\bot\}^b$ replaces
each occurrence of $\bot$ by either zero or one, while
preserving all pure bits.

\begin{lemma}[Gray-code extraction gadget]
\label{lem:gray-gadget}
For every $b\geq 1$, one can construct in $O(b)$ time an
acyclic \PureCircuit\ gadget
\[
    E_b(r;g_1,\ldots,g_b)
\]
with one external input $r$, and $b$ output
vertices. In every satisfying assignment, its output 
$\mathbf{g}=(g_1,\ldots,g_b)$ satisfies:
\begin{enumerate}
    \item At most one entry of $\mathbf{g}$ is impure.

    \item For any two pure completions
    $\widehat{\mathbf{g}}$ and $\widetilde{\mathbf{g}}$,
    \[
      \bigl|D_b(\widehat{\mathbf{g}})
            -D_b(\widetilde{\mathbf{g}})\bigr|\leq 1.
    \]

    \item If $r=0$, every pure completion decodes to zero.
    If $r=1$, every pure completion decodes to $2^b-1$.
\end{enumerate}
\end{lemma}

\begin{proof}
We first construct a constant-size acyclic gadget
$S(x;\ell,h)$. Apply $\PURIFY(x;u,v)$ and set
\[
 \ell=\AND(u,v),\qquad h=\OR(u,v).
\]
If $x$ is pure, then $\ell=h=x$. Moreover,
\begin{equation}
    \label{eq:S logic}
    \ell=\bot\Longrightarrow h=1,
 \qquad
 h=\bot\Longrightarrow \ell=0.
\end{equation}
Indeed, if $\ell=\bot$, neither $u$ nor $v$ can be zero.
Since at least one of them is pure, at least one must be
one, forcing $h=1$. The other implication follows
symmetrically.

Using two copies of $S$, define a constant-size gadget
$G(x;g,t)$ by
\[
 S(x;a,s),\qquad S(s;g,c),\qquad
 t=\AND(c,\NOT(a)).
\]
A pure input $x=0$ forces $(g,t)=(0,0)$, while $x=1$
forces $(g,t)=(1,0)$. Furthermore,
\[
 g=\bot\Longrightarrow t=1.
\]
Since a pure
input to $S$ forces both outputs to be pure, $g=\bot$ implies that $s=\bot$. 
By \Cref{eq:S logic}, $g = \bot$ forces $c=1$, and $s = \bot$ forces $a=0$.
Thus $t=1$.

Set $t_0=r$ and, for each $j\in[b]$, introduce a gadget
\[
 G(t_{j-1};g_j,t_j).
\]
The resulting gadget is acyclic. Each step uses constantly
many gates. This circuit can be constructed in $O(b)$ time 
and is our Gray-code extraction gadget $E_b(r;g_1,\ldots,g_b)$.
We now verify it satisfies the three required properties.

Fix any satisfying assignment. If $g_j=\bot$, then $t_j=1$.
When $j<b$, the next step therefore gives
$(g_{j+1},t_{j+1})=(1,0)$, and every subsequent step gives
$(g_k,t_k)=(0,0)$. Thus, an impure output cannot be followed
by another impure output, proving the first property.
It also follows that the suffix after an impure bit is
$10 \ldots 0$.

If the output string is pure, the second property is immediate. 
Otherwise, let $j$ be its unique impure position, and first suppose that $j<b$. The two pure completions
agree before position $j$, for example,
\[
\widehat{\mathbf{g}}=P\,0\,1\,0^{b-j-1}, 
\qquad \widetilde{\mathbf{g}} = P\,1\,1\,0^{b-j-1}, 
\]
where $P$ is their common Gray-code prefix. 
Using the decoding circuit $D_b$, their binary representations have the forms, for example, 
\[
B\,0\,1^{b-j}, \qquad B\,1\,0^{b-j},
\]
where $B$ is their common binary prefix. 
The two displayed suffixes represent
$2^{b-j}-1$ and $2^{b-j}$, respectively. Hence the decoded
integers differ by exactly one. In the remaining case $j=b$,
the decoded binary representations agree in their first $b-1$
bits and differ only in their last bit, so the decoded integers
again differ by exactly one.

Finally, if $r=0$, every step outputs zero, giving the
Gray word $0^b$, which decodes to zero. If $r=1$, the first
step outputs one and passes state zero to the next step.
The resulting Gray word is $10^{b-1}$, whose decoded
binary representation is $1^b$, representing $2^b-1$.
This proves the third property.
\end{proof}

We now describe the reduction from \TwoDStrongSperner\ to
\BFSPureCircuit$(5)$. We follow the four-stage framework
of~\cite{DFHM24b}: purification, circuit evaluation,
sorting, and selection. 
We begin by creating two vertices $r_1,r_2$, called the \emph{original inputs},
one for each coordinate.
Each is the input to a Gray-code extraction gadget and receives 
feedback from the selection stage.
\begin{enumerate}
    \item \emph{Purification stage.}
    Use $\PURIFY$ trees to produce five copies of the resulting
    Gray strings. At least three copies are fully pure, and
    all fully pure copies decode to points in a common
    unit grid cell.

    \item \emph{Circuit stage.}
    For each copy, decode Gray strings to binary number and
    evaluate the labeling circuit $C$. Fully pure copies
    yield the correct labels; the outputs associated with
    the remaining copies may be arbitrary.

    \item \emph{Sorting stage.}
    For each coordinate, weakly sort the five corresponding
    label bits using $\AND$ and $\OR$ gates. If at least
    three inputs have the same pure value, the third
    output is forced to that value.

    \item \emph{Selection stage.}
    Select the third sorted output for each coordinate
    and feed it back to the \emph{original input} through
    a chain of $\COPY$ gates.  
\end{enumerate}

\begin{theorem}
\label{thm:ordered-source}
\BFSPureCircuit$(5)$ is \PPAD-hard.
\end{theorem}

\begin{proof}
We reduce from \TwoDStrongSperner.
Let $C$ be the labeling circuit of the source instance,
and write $L=2^b-1$.
Create two original input vertices $r_1,r_2$, and apply
\Cref{lem:gray-gadget} to obtain
\[
 E_b(r_i;g_{i,1},\ldots,g_{i,b}),
 \qquad i\in[2].
\]
Throughout the proof, let $\mathbf{x}$ denote an arbitrary
satisfying assignment of the circuit constructed below.

For each $(i,j)\in[2]\times[b]$, construct a binary tree
of $\PURIFY$ gates rooted at $g_{i,j}$ and five leaves
$g_{i,j}^{(1)},\ldots,g_{i,j}^{(5)}$.
Define the set of good copies by
\[
 \mathcal G
 :=
 \left\{
 k\in[5]:
 \mathbf{x}[g_{i,j}^{(k)}]\in\{0,1\}
 \text{ for all }(i,j)\in[2]\times[b]
 \right\}.
\]
The argument for $\PURIFY$ trees
in~\cite[Lemma~3.2]{DFHM24b} shows that pure
inputs are preserved at every leaf, while an impure
input can yield at most one impure leaf.
By \Cref{lem:gray-gadget}, at most two of the $2b$ tree
inputs are impure. Hence at most two copies are not good,
and $|\mathcal G|\geq 3$.

For $k\in\mathcal G$, let $p^{(k)}$ be the integer point
obtained by decoding its two Gray strings.
Each string is a pure completion of the corresponding
extractor output, so \Cref{lem:gray-gadget} gives
\[
 \|p^{(k)}-p^{(k')}\|_\infty\leq 1
 \qquad\text{for all }k,k'\in\mathcal G.
\]

For every $k\in[5]$, construct a copy of the Boolean circuit
that decodes two Gray strings into binary and then evaluates
$C$, implementing it with $\NOT$ and $\NOR$ gates.
Connect its inputs to the vertices $g_{i,j}^{(k)}$, and denote
its two output vertices by $v_1^{(k)},v_2^{(k)}$.
The proof of~\cite[Lemma~3.3]{DFHM24b} applies
to this composed Boolean circuit and yields
\[
 \mathbf{x}[v_i^{(k)}]=C_i(p^{(k)})
 \qquad\text{for all }k\in\mathcal G,\ i\in[2].
\]

For each $i\in[2]$, apply a fixed five-input sorting
network to $v_i^{(1)},\ldots,v_i^{(5)}$, replacing each
comparator by $\AND$ and $\OR$ gates for its lower and
upper outputs, respectively.
Denote the resulting outputs by
$w_i^{(1)},\ldots,w_i^{(5)}$.
By the weak-sorting property proved
in~\cite[Lemma~3.4]{DFHM24b}, if at least three
inputs equal a pure value $s$, then
$\mathbf{x}[w_i^{(3)}]=s$.

The selection stage needs some additional work.
For each $i\in[2]$, select $q_i:=w_i^{(3)}$.
We construct the feedback chains from $q_i$ to $r_i$
together with the final gate ordering.
Before adding feedback, the construction is acyclic,
with external inputs $r_1,r_2$ and selected outputs
$q_1,q_2$.
Expand all auxiliary gates into basic gates, and let
$h_1,\ldots,h_s$ be a reverse topological ordering.
Every wire between these original gates points backward.
We will place the gates producing $r_1,r_2$ after all
original gates, and hence after the gates producing
$q_1,q_2$.

We construct the final ordering $\mathcal L$ by scanning
the fixed sequence $h_1,\ldots,h_s$ from left to right.
Initially, $\mathcal L$ is empty and both feedback chains
are inactive. For each $j\in[s]$, proceed as follows:
\begin{enumerate}
    \item Append $h_j$ to $\mathcal L$.

    \item For each $i\in[2]$ such that $h_j$ produces
    $q_i$, activate feedback chain $i$ and set its
    current endpoint to $z_i=q_i$.

    \item For each active chain, in the order $i=1,2$,
    introduce a $\COPY(z_i;z_i')$ gadget. Append its
    two basic gates
    \[
      \NOT(z_i;a_{i,j}),\qquad
      \NOT(a_{i,j};z_i')
    \]
    to $\mathcal L$ in this order, and update the
    current endpoint to $z_i=z_i'$.
    All new vertices are fresh, except that in the
    final iteration $j=s$, the output $z_i'$ is $r_i$.
\end{enumerate}
Thus, after each original gate, we insert one $\COPY$
gadget for each feedback chain activated so far.
Both chains are active by the final iteration, and
their last gadgets produce $r_1,r_2$.

We now verify the forward spans are bounded in order $\mathcal L$.
Wires between original gates remain backward.
Within a $\COPY$ gadget, the span is one.
Between consecutive $\COPY$ gadgets on the same
chain, the connecting wire crosses at most two gates
of the other chain and one original gate, so its span
is at most four. The wire from the gate producing $q_i$
to the first gate of its first $\COPY$ gadget has span at
most three, since at most the two gates of the other active
chain are inserted between them. Finally, every wire from
$r_i$ to an original gate points backward, because the last
$\COPY$ gadget producing $r_i$ appears after all original
gates. These cases cover every wire, so every forward span
is at most five.

Each vertex has a unique producing gate, and at most
four relay gates are inserted after each original gate.
Thus, the construction remains polynomial in size.
Finally, since every $\COPY$ gadget preserves pure values,
the feedback chains satisfy
\[
 \mathbf{x}[q_i]=s\in\{0,1\}
 \Longrightarrow
 \mathbf{x}[r_i]=s,
 \qquad i\in[2].
\]

It remains to check the correctness.
Suppose, for contradiction, there are three good copies
whose $i$-th labels all equal $s$ for some $i \in \{1,2\}$.
The circuit stage therefore produces at least three
labels equal to $s$. Sorting and selection imply
$\mathbf{x}[q_i]=\mathbf{x}[r_i]=s$.
By the endpoint property of \Cref{lem:gray-gadget},
every good copy then has $p_i^{(k)}=s(2^b - 1)$.
The boundary conditions of \TwoDStrongSperner\ require
$C_i(p^{(k)})=1-s$, leading to a contradiction.

Consider three good copies such that, in each coordinate,
the labels of their decoded points are not all equal.
Choose any two of these
points. If their labels differ in both coordinates, they
form the desired pair. If their labels agree in both
coordinates, choose either of them together with the
remaining point. Otherwise, their labels differ in exactly
one coordinate. In this case, choose the remaining point
together with the one of the first two points whose label
differs from its label in that coordinate.
\end{proof}

\section{\PPAD-hardness for ABE-B}
\label{sec:budget}
\label{sec:ABE-B}

Fix a constant $0\leq\gamma<1$, and let $\varepsilon=1-\gamma>0$.
Let $g_1,\ldots,g_m$ be an instance of the ordered problem in
\Cref{thm:ordered-source}, on vertex set $V$, with $|V|=n$.
Let $d=5$, and set
\begin{equation*}
 a=\left\lceil\frac{2}{\varepsilon}\right\rceil,\qquad c=a^{36},\qquad
 D=4d+2,\qquad
 k=a^{12}c^D.
\end{equation*}
The parameters $a$, $c$, $D$, and $k$ are constants independent of
the instance, depending only on $\gamma$. In particular,
$a\geq2$ and $\varepsilon a\geq2$.
Every exponent of $c$ below is an integer multiple of $1/36$.
Consequently all values and budgets are integers.

\subsection{Construction of the auction}

We are given an instance of \BFSPureCircuit$(5)$, including an
ordering of its gates with forward span at most five. We introduce
one real bidder for each vertex of the circuit. Thus, for
simplicity, we slightly abuse notation by using $v\in V$ to denote
both a circuit vertex and its corresponding bidder; the intended
meaning will be clear from context. Budgets of all real bidders
are initially zero.

We organize the construction into levels, with level $\ell$ having
basic value scale $c^\ell$. Each level corresponds to one gate output.
A $\NOT$ or $\NOR$ gate $g_I$ uses level $4I+1$, whereas an
$\NPURIFY$ gate $g_I$ uses levels $4I+1$ and $4I+3$, one for
each output. The input and output bidders associated with these
levels are specified below:
\[
\begin{array}{c|c|c|c}
\text{gate }g_I & \text{level} & \text{input bidders}
  & \text{output bidder}\\ \hline
\NOT(u;w)       & 4I+1 & u   & w\\
\NOR(u,v;w)     & 4I+1 & u,v & w\\
\NPURIFY(u;v,w) & 4I+1 & u   & v\\
               & 4I+3 & u   & w
\end{array}
\]
We call these the \emph{occupied levels}. Distinct occupied
levels differ by at least two.

At each occupied level $\ell$, we create items and increase
budgets according to the following common rules.
\begin{itemize}
\item \emph{Output bidder.}
Give the output bidder $k$ items, each of value $c^{\ell+2/3}$,
and increase its budget by $kc^{\ell+1/12}$.
These $k$ items form its \emph{output block}.

\item \emph{Input bidders.}
Give each input bidder listed at this level one item of value
\[
\begin{cases}
c^{\ell+1/3}, & \ell=4I+1,\\
c^{\ell+2/9}, & \ell=4I+3,
\end{cases}
\]
and increase its budget by $c^{\ell-1}$.
If the two inputs $u$ and $v$ of a $\NOR$ gate coincide,
their common bidder receives two items and two budget increments.
\end{itemize}
All budget increments are cumulative, with repeated uses of
a vertex as an input counted separately.

Thus, the output construction and the budget increment for
each input item follow the same rule at every level.
The second level of an $\NPURIFY$ gate uses the smaller exponent
offset $2/9$ for its input item. This difference allows its two output blocks to respond
to different ranges of the input multiplier, ensuring that
at least one output is pure, as shown in the correctness proof.

Finally introduce a guard bidder with one item of value $c^t$
for each $t\in[T]$, where $T=4m+4$, and budget $aTc^T$ (this is a sufficiently large number).
Set the number of slots equal to the total number of items
and all click-through rates to one. This completes the construction.

For a real bidder $v\in V$, let $\ell_v$ be its output level.
The vertex $v$ is an output of exactly one gate, so it receives
exactly one output block: $k$ items, each of value
$c^{\ell_v+2/3}$, together with a budget increment of
$kc^{\ell_v+1/12}$.
Once the common parameters $c$ and $k$ are fixed, these
contributions depend only on $\ell_v$. The remaining items
and budget increments depend on which gates use $v$ as an
input, with repeated uses counted separately.

Let $B_v^{\mathrm{in}}$ be the sum of the budget increments $v$ receives
when used as a gate input. Its total budget is therefore
\[
B_v=kc^{\ell_v+1/12}+B_v^{\mathrm{in}}.
\]
Similarly, let $V_v^{\mathrm{in}}$ be the sum of the values of the items
$v$ receives when used as a gate input. The total value of
its items is then $kc^{\ell_v+2/3}+V_v^{\mathrm{in}}$, where the first term
is the value of its output block.

\subsection{Multiplier bounds and local prices}

Given a \PureCircuit\ instance with a gate ordering of forward span
at most five, we construct the auction described above using the
parameters determined by the fixed $\gamma$. Fix any
$\gamma$-approximate ABE-B $(\alpha,\pi)$ of this auction.

Denote the guard bidder by $\star$, and let $P_h$ be bidder $h$'s
expected total payment under $(\alpha,\pi)$.
The equilibrium conditions give, for every bidder $h\in V\cup\{\star\}$,
\[
 0\leq\alpha_h\leq1,\qquad P_h\leq B_h,\qquad
 P_h<(1-\gamma)B_h\ \Longrightarrow\ \alpha_h=1.
\]

For a real bidder $v$, write
$P_v=P_v^{\mathrm{out}}+P_v^{\mathrm{in}}$, where the two terms
are its expected payments on its output block and on the items
it receives when used as a gate input, respectively. The latter
items have total value $V_v^{\mathrm{in}}$, so
\begin{equation}\label{eq:budget-payment-split}
 0\leq P_v^{\mathrm{in}}\leq\alpha_v V_v^{\mathrm{in}}\leq V_v^{\mathrm{in}}.
\end{equation}
We first show that the guard bids at its values. We then bound
$V_v^{\mathrm{in}}$ and the real bidders' multipliers, before comparing bids
at different levels.

\begin{lemma}
\label{lem:budget-guard}
The guard bidder has multiplier $\alpha_\star=1$.
\end{lemma}
\begin{proof}
The guard has one item of each value $c^t$, for $t\in[T]$,
and budget $B_\star=aTc^T$. Since $\alpha_\star\leq1$,
the payment bound above gives
\[
 P_\star\leq\sum_{t=1}^T\alpha_\star c^t
 \leq Tc^T
 <(1-\gamma)aTc^T=(1-\gamma)B_\star.
\]
The strict inequality uses $(1-\gamma)a=\varepsilon a\geq2$.
Approximate maximal pacing therefore forces $\alpha_\star=1$.
In particular, its item of value $c^t$ has bid exactly $c^t$
for every $t\in[T]$.
\end{proof}

The next claim shows that every real bidder $v\in V$ with
$\alpha_v<1$ must have expected payment from its output block
strictly greater than $kc^{\ell_v}$.

We would like to highlight that this is where we use the bounded forward span of the
\PureCircuit\ instance. It ensures that all items bidder $v$
receives as a gate input lie at levels at most $\ell_v+D$, where
$D=4d+2$ is a constant. Thus, even if $v$ is used as an input
by many gates, their total value satisfies
\[
 V_v^{\mathrm{in}}\leq2\bigl(c^{1/3}+c^{1+1/3}+\cdots+c^{\ell_v+D+1/3}\bigr)
 =O_\gamma(c^{\ell_v}).
\]
The implicit constant depends only on $\gamma$ and is
independent of the circuit size. This allows us to choose the
number $k$ of items in each output block independently of the
circuit size as well.

\begin{claim}
\label{clm:budget-slack}
For every real bidder $v\in V$,
\begin{equation}\label{eq:budget-slack}
 kc^{\ell_v}+V_v^{\mathrm{in}}
 <(1-\gamma)kc^{\ell_v+1/12}
 \leq(1-\gamma)B_v.
\end{equation}
\end{claim}
\begin{proof}
Fix a real bidder $v$, and let $g_I$ be the unique gate of which
$v$ is an output. Its output level satisfies $\ell_v=4I+1$ or $\ell_v=4I+3$.
Consider any gate $g_J$ that uses $v$ as an input. If $J>I$,
the forward-span bound gives $J-I\leq d$; if $J<I$, this
inequality holds automatically. Every item assigned to $v$
by this gate lies at a level $\ell\leq4J+3$, so
\[
 \ell-\ell_v\leq4(J-I)+2\leq4d+2=D.
\]

At each occupied level $\ell$, bidder $v$ receives at most two
items as a gate input, even if it fills both input positions
of a $\NOR$ gate. Each such item has value at most $c^{\ell+1/3}$.
Since all these items lie at levels at most $\ell_v+D$,
their total value satisfies
\[
 V_v^{\mathrm{in}}\leq2\sum_{t=0}^{\ell_v+D}c^{t+1/3}
 <\frac{2c^{1/3}}{1-c^{-1}}c^{\ell_v+D}.
\]
Since $c^{1/3}=a^{12}$ and $a\geq2$, we have
$2/(1-c^{-1})=2/(1-a^{-36})<3$. Hence
\begin{equation*}
 V_v^{\mathrm{in}}<3a^{12}c^{\ell_v+D}=3kc^{\ell_v}.
\end{equation*}
Recall that the budget increment associated with the output block
is $kc^{\ell_v+1/12}$.
The claim compares this with the payment benchmark $kc^{\ell_v}$ over the block.
After multiplying the budget contribution by $1-\gamma$, as
required by approximate maximal pacing, the difference from this
payment benchmark is $kc^{\ell_v}\Delta$, where
\[
 \Delta\coloneqq(1-\gamma)c^{1/12}-1
 =\varepsilon a^3-1\geq2a^2-1\geq7.
\]
Here we used $\varepsilon a\geq2$, $a\geq2$, and $c^{1/12}=a^3$.
By construction, $k=a^{12}c^D$, so the upper bound on $V_v^{\mathrm{in}}$
proved above is $3kc^{\ell_v}$. Combining this with $\Delta\geq7$ gives
\[
 kc^{\ell_v}\Delta\geq7kc^{\ell_v}>3kc^{\ell_v}>V_v^{\mathrm{in}}.
\]
Rearranging proves the strict inequality in
\eqref{eq:budget-slack}; the remaining inequality follows from
$B_v^{\mathrm{in}}\geq0$.
\end{proof}

\begin{lemma}
\label{lem:budget-floor}
Every real bidder $v\in V$ has multiplier $\alpha_v\in[c^{-1},1]$.
\end{lemma}
\begin{proof}
Fix $v\in V$ and suppose, for a contradiction, that
$\alpha_v<c^{-1}$. Each of its $k$ output items has bid
$\alpha_v c^{\ell_v+2/3}<c^{\ell_v-1/3}<c^{\ell_v}$.
Thus $P_v^{\mathrm{out}}<kc^{\ell_v}$. Combining
\eqref{eq:budget-payment-split} with \Cref{clm:budget-slack} gives
\[
 P_v=P_v^{\mathrm{out}}+P_v^{\mathrm{in}}
 <kc^{\ell_v}+V_v^{\mathrm{in}}<(1-\gamma)B_v.
\]
Approximate maximal pacing would require $\alpha_v=1$,
contradicting $\alpha_v<c^{-1}$. 
\end{proof}

For the output block at an occupied level $\ell$, its
\emph{corresponding input items} are the items assigned to the
input bidders listed at that level in the construction.
Their owners differ from the output bidder, since a gate's
outputs are distinct from its inputs.

\begin{lemma}
\label{lem:budget-locality}
At every occupied level $\ell$, the following hold:
\begin{enumerate}
\item Every real bid lies strictly between the guard bids
$c^{\ell-1}$ and $c^{\ell+1}$.
\item Every input item pays at least $c^{\ell-1}$, the budget
increment assigned with that item. Summing over all levels gives
$P_v^{\mathrm{in}}\geq B_v^{\mathrm{in}}$ for every real bidder $v$.
\item The output items satisfy the following payment bounds:
\begin{enumerate}[label=(\roman*)]
\item If all input bids are at most $c^\ell$, every output
item pays at most $c^\ell$.
\item If an input bid $b$ is strictly below the output bid,
every output item pays at least $b$.
\end{enumerate}
\end{enumerate}
\end{lemma}
\begin{proof}
Fix an occupied level $\ell$ and an allocation $\sigma$ in the
support of $\pi$. For (1), every real item at this level has value
$c^{\ell+\theta}$ for some $\theta\in\{2/9,1/3,2/3\}$.
By \Cref{lem:budget-floor}, its bid
belongs to $[c^{\ell-1+\theta},c^{\ell+\theta}]$, which is strictly inside
$(c^{\ell-1},c^{\ell+1})$. Distinct occupied levels differ by at least two,
so real items at lower levels have bids below $c^{\ell-1}$,
and real items at higher levels have bids above $c^{\ell+1}$.
By \Cref{lem:budget-guard}, the guard bids are exactly its item
values.

For (2), the guard item at $c^{\ell-1}$ is strictly below every input item at
level $\ell$ and has a different owner. Thus, every input item at
level $\ell$ pays at least $c^{\ell-1}$, exactly the budget increment
assigned to it in the construction, and therefore $P_v^{\mathrm{in}}\geq B_v^{\mathrm{in}}$.

For (3), the real items at level $\ell$ consist only of the output block
and its corresponding input items. The output items all have
the same owner, so they cannot set one another's prices.
By the bid bounds above, every item outside this level that
can determine an output item's price bids at most $c^\ell$.
Thus, if all input bids are also at most $c^\ell$, every output
item pays at most $c^\ell$, proving (3)(i).

For (3)(ii), the input item with bid $b$ lies strictly below every
output item and has a different owner. Hence every output item
pays at least $b$. Both conclusions hold for every $\sigma$
in the support of $\pi$.
\end{proof}

\subsection{Decoding and proof of correctness}

For the fixed equilibrium $(\alpha,\pi)$, define an assignment
$\xx:V\to\{0,1,\bot\}$ using only the real bidders' multipliers:
\begin{equation}\label{eq:budget-decoding}
 \xx[v]=
 \begin{cases}
 0,&\alpha_v\leq c^{-1/3},\\
 1,&\alpha_v=1,\\
 \bot,&\text{otherwise}.
 \end{cases}
\end{equation}
Thus multipliers at most $c^{-1/3}$ represent zero, the maximum
multiplier $1$ represents one, and intermediate multipliers
represent $\bot$. We show that this assignment satisfies every
gate of the given circuit.

\begin{lemma}
\label{lem:budget-gates}
The assignment in \eqref{eq:budget-decoding} satisfies every gate
of the given \BFSPureCircuit$(5)$ instance.
\end{lemma}
\begin{proof}
We first derive two consequences of \Cref{lem:budget-locality}.
Fix an occupied level $\ell$, and let $v$ be its output bidder,
so $\ell_v=\ell$. For the input and output items at this level:
\begin{enumerate}[label=(\roman*)]
\item If all input bids are at most $c^\ell$, then
$\alpha_v=1$, and hence $\xx[v]=1$.
\item If an input bid $b$ exceeds $c^{\ell+1/12}$, then
the output bid $\alpha_v c^{\ell+2/3}$ is at most $b$.
\end{enumerate}

For (i), each of the $k$ output items pays at most $c^\ell$
by part (3)(i) of \Cref{lem:budget-locality}. The remaining
items have expected payment at most $V_v^{\mathrm{in}}$ by
\eqref{eq:budget-payment-split}. Therefore
\[
 P_v=P_v^{\mathrm{out}}+P_v^{\mathrm{in}}
 \leq kc^\ell+V_v^{\mathrm{in}}<(1-\gamma)B_v,
\]
where the strict inequality follows from \Cref{clm:budget-slack}.
Approximate maximal pacing forces $\alpha_v=1$.

For (ii), suppose that the output bid is strictly greater than
$b$. Part (3)(ii) of \Cref{lem:budget-locality} implies that
every output item pays at least $b$, and part (2) gives
$P_v^{\mathrm{in}}\geq B_v^{\mathrm{in}}$. Hence
\[
 P_v=P_v^{\mathrm{out}}+P_v^{\mathrm{in}}
 \geq kb+B_v^{\mathrm{in}}>kc^{\ell+1/12}+B_v^{\mathrm{in}}=B_v,
\]
contradicting budget feasibility. Thus the output bid is at
most $b$; equality is allowed.

We now apply (i) and (ii) to verify the gate constraints.

\paragraph{$\NOT$ and $\NOR$ gates.}
Fix any $\NOT$ or $\NOR$ gate $g_I$ and set $\ell=4I+1$.
Every input item at this level has value $c^{\ell+1/3}$.
If an input vertex is assigned zero, its multiplier is at most
$c^{-1/3}$, so its input bid is at most $c^\ell$.
Thus (i) forces output one for a $\NOT$ gate with input zero
and for a $\NOR$ gate with both inputs zero.

If an input vertex is assigned one, its multiplier is exactly
one and its input bid is $c^{\ell+1/3}>c^{\ell+1/12}$.
By (ii), the output bid is at most $c^{\ell+1/3}$. Since each
output item has value $c^{\ell+2/3}$, its owner's multiplier
is at most $c^{-1/3}$, so the output is assigned zero.
For a $\NOR$ gate, this holds whenever either input is one,
regardless of the other input. These are exactly the required
constraints.

\paragraph{$\NPURIFY$ gates.}
Fix any gate $g_I=\NPURIFY(u;v,w)$ and set $\ell=4I+1$.
The output blocks of $v$ and $w$ lie at levels $\ell$ and
$\ell+2$, respectively. The input items at these two levels
belong to $u$ and have bids
\[
 b^{(1)}=\alpha_u c^{\ell+1/3},
 \qquad b^{(2)}=\alpha_u c^{\ell+2+2/9}.
\]
The key inequality behind these two input offsets is
\[
 \frac13-\frac29=\frac19>\frac1{12}.
\]
Thus, at the threshold $\alpha_u=c^{-2/9}$, the second input
bid equals its level's benchmark $c^{\ell+2}$, while the first
input bid exceeds the budget per output item $c^{\ell+1/12}$.
This allows (i) and (ii) to cover every possible input multiplier.

We first show that at least one output is pure.
If $\alpha_u\geq c^{-2/9}$, then
$b^{(1)}\geq c^{\ell+1/9}>c^{\ell+1/12}$.
By (ii), the output bid $\alpha_v c^{\ell+2/3}$ is at most
$b^{(1)}$, so
\[
 \alpha_v\leq\alpha_u c^{-1/3}\leq c^{-1/3}.
\]
Hence $\xx[v]=0$. If $\alpha_u\leq c^{-2/9}$, then
$b^{(2)}\leq c^{\ell+2}$, so (i) at level $\ell+2$ gives
$\alpha_w=1$ and $\xx[w]=1$.
Every input multiplier satisfies at least one of these two
inequalities. At equality, both conclusions hold.

It remains to check the constraints when the input is pure.
If $\xx[u]=0$, then $\alpha_u\leq c^{-1/3}<c^{-2/9}$.
The first input bid is at most $c^\ell$, and the second is at
most $c^{\ell+2}$. By (i), both outputs are one.
If $\xx[u]=1$, then $\alpha_u=1$, and the preceding argument
already gives $\xx[v]=0$. At level $\ell+2$, the input bid is
$c^{\ell+2+2/9}>c^{\ell+2+1/12}$. By (ii), the output bid
$\alpha_w c^{\ell+2+2/3}$ is at most
$c^{\ell+2+2/9}<c^{\ell+2+1/3}$, so
$\alpha_w\leq c^{-1/3}$ and $\xx[w]=0$.
These are exactly the $\NPURIFY$ constraints.
\end{proof}

\begin{theorem}[Restatement of \Cref{thm:ppad-results}, ABE-B hardness]
\label{thm:budget-hardness}
For every fixed $0\leq\gamma<1$, computing a $\gamma$-approximate
ABE-B is \PPAD-hard even when
the number of slots equals the number of items and all
click-through rates equal one.
\end{theorem}
\begin{proof}
The construction has $kn$ items in the real bidders' output
blocks, $O(m)$ items assigned to gate inputs, and $4m+4$ guard
items. Since $k$ depends only on the fixed $\gamma$, the total
number of explicitly listed items is $O_\gamma(n+m)$.
All powers of $a$ used in item values and real bidders' budget
increments have integer exponents $O(m)$, so these quantities
have $O_\gamma(m)$ bits. Summing the budget increments preserves
polynomial encoding length, as does the guard budget $aTc^T$.
The auction can therefore be constructed in polynomial time.

For any $\gamma$-approximate ABE-B $(\alpha,\pi)$ of this
auction, \Cref{lem:budget-gates} shows that
\eqref{eq:budget-decoding} gives a solution of the source
instance. The decoding takes polynomial time and uses only
$\alpha$; the proof applies to every valid distribution $\pi$.
The claimed hardness follows from \Cref{thm:ordered-source}.
\end{proof}

\section{\PPAD-hardness for ABE-R}
\label{sec:roi}
\label{sec:ABE-R}

Fix a constant $0\leq\gamma<1$, and let $\varepsilon=1-\gamma>0$.
We are given an instance of \BFSPureCircuit$(5)$ on vertex
set $V$, with $|V|=n$, including a gate ordering
$g_1,\ldots,g_m$ of forward span at most $d=5$.
We use the same parameter notation as in the budget construction,
but now set
\begin{equation*}
 a=\left\lceil\frac{2}{\varepsilon}\right\rceil,\qquad c=a^{128},\qquad
 D=2d+1,\qquad
 k=4a^{17}c^D.
\end{equation*}
The parameters $a$, $c$, $D$, and $k$ are constants independent of
the instance, depending only on $\gamma$. In particular,
$a\geq2$ and $\varepsilon a\geq2$.
All item values below are integer powers of $a$ and hence are
integers. The construction uses two guard bidders. We will show
that ROI feasibility bounds both guard multipliers between $1$
and $a$. These bounds will allow us to compare bids at different
levels even when the guards have different multipliers.

\subsection{Construction of the auction}

We introduce one real bidder for each vertex of the circuit. Thus, for
simplicity, we slightly abuse notation by using $v\in V$ to denote
both a circuit vertex and its corresponding bidder; the intended
meaning will be clear from context.

As in the budget construction, we organize the items into
levels, with level $\ell$ having basic value scale $c^\ell$.
A $\NOT$ or $\NOR$ gate $g_I$ uses level $2I$, whereas an
$\NPURIFY$ gate $g_I$ uses levels $2I$ and $2I+1$, one for
each output. The input and output bidders are specified below:
\[
\begin{array}{c|c|c|c}
\text{gate }g_I & \text{level} & \text{input bidders}
  & \text{output bidder}\\ \hline
\NOT(u;w)       & 2I   & u   & w\\
\NOR(u,v;w)     & 2I   & u,v & w\\
\NPURIFY(u;v,w) & 2I   & u   & v\\
               & 2I+1 & u   & w
\end{array}
\]
We call these the \emph{occupied levels}. Each occupied level
has exactly one output bidder.

At each occupied level $\ell$, we create items according to
the following common rules.
\begin{itemize}
\item \emph{Output bidder.}
Give the output bidder $k$ items, each of value $a^{64}c^\ell$.
These $k$ items form its \emph{output block}.

\item \emph{Input bidders.}
Give each input bidder listed at this level one item of value
\[
\begin{cases}
a^{16}c^\ell, & \ell=2I,\\
a^4c^\ell, & \ell=2I+1.
\end{cases}
\]
If the two inputs $u$ and $v$ of a $\NOR$ gate coincide,
their common bidder receives two items.
\end{itemize}
All items assigned by these rules are included, with repeated
uses of a vertex as an input counted separately.
The output items follow the same rule at every level.
The second level of an $\NPURIFY$ gate uses the smaller
coefficient $a^4$ for its input item. This difference allows
its two output blocks to respond to different ranges of the
input multiplier, ensuring that at least one output is pure,
as shown in the correctness proof.

Finally introduce a set $\mathcal G$ of two guard bidders,
each with one item of every value in the list
\begin{equation}\label{eq:roi-guard-catalog}
 c,c^2,\ldots,c^{2m+2},\quad ac^{2m+2}.
\end{equation}
The two identical lists allow us to compare the guards'
payments with their values. The additional item of value
$ac^{2m+2}$ will be used to bound their multipliers in
\Cref{lem:roi-guards}.
Set the number of slots equal to the total number of items
and all click-through rates to one. We use the multiplier
cap from the preliminaries, namely one plus the ratio of the
sum of all item values to the smallest item value, and denote
it by $A_0$. We set $A=A_0$.
Since the guards have items of values $c$ and $c^2$, we have
$A_0>c>a^{65}$. This completes the construction.

For a real bidder $v\in V$, let $\ell_v$ be its output level.
The vertex $v$ is an output of exactly one gate, so it receives
exactly one output block: $k$ items, each of value
$a^{64}c^{\ell_v}$. Once $a$, $c$, and $k$ are fixed, the
value of this block depends only on $\ell_v$.
The remaining items depend on which gates use $v$ as an input,
with repeated uses counted separately.

Let $V_v^{\mathrm{in}}$ be the sum of the values of the items
$v$ receives when used as a gate input. Its total value is
therefore
\[
 V_v=ka^{64}c^{\ell_v}+V_v^{\mathrm{in}}.
\]
Since every item is allocated with click-through rate one,
this is also the value bidder $v$ receives in every allocation.
As in the budget construction, we will use the forward-span
restriction to bound $V_v^{\mathrm{in}}$ relative to the output block.

\subsection{Multiplier bounds and local prices}

Given a \PureCircuit\ instance with a gate ordering of forward span
at most five, we construct the auction described above using the
parameters determined by the fixed $\gamma$. Fix any
$\gamma$-approximate ABE-R $(\alpha,\pi)$ of this auction.

Let $P_h$ be bidder $h$'s expected total payment under $(\alpha,\pi)$,
and let $V_h$ be the total value of its items. Since every item is
allocated with click-through rate one, this value is independent
of the allocation. The equilibrium conditions give, for every bidder
$h\in V\cup\mathcal G$,
\[
 1\leq\alpha_h\leq A_0,\qquad P_h\leq V_h,\qquad
 P_h<(1-\gamma)V_h\ \Longrightarrow\ \alpha_h=A_0.
\]

For a real bidder $v$, write
$P_v=P_v^{\mathrm{out}}+P_v^{\mathrm{in}}$, where the two terms
are its expected payments on its output block and on the items
it receives when used as a gate input, respectively. The latter
items have total value $V_v^{\mathrm{in}}$, so
\[
 0\leq P_v^{\mathrm{in}}\leq\alpha_v V_v^{\mathrm{in}}.
\]
Write $\lambda=\max_{h\in\mathcal G}\alpha_h$ for the largest
guard multiplier. For each $t\in[2m+2]$, the larger guard bid
on an item of value $c^t$ is therefore $\lambda c^t$.
We first bound the guard multipliers using only ROI feasibility.
We then bound $V_v^{\mathrm{in}}$ and the real bidders' multipliers,
before comparing bids at different levels.

\begin{lemma}
\label{lem:roi-guards}
Every guard bidder $h\in\mathcal G$ has multiplier
$\alpha_h\in[1,a]$.
\end{lemma}
\begin{proof}
The lower bound is part of the multiplier domain. Each guard
has the same total value
\[
 W=\sum_{t=1}^{2m+2}c^t+ac^{2m+2},
\]
so ROI feasibility gives $P_h\leq W$ for each $h\in\mathcal G$.
It remains to prove $\lambda\leq a$.

Suppose first that both guards have multiplier $\lambda$.
Fix an allocation $\sigma$ in the support of $\pi$.
For every value $q$ in \eqref{eq:roi-guard-catalog}, the two
guard items of value $q$ both bid $\lambda q$.
Whichever of these items appears first in $\sigma$ pays at
least $\lambda q$, because the other item has a different
owner and the same bid. Summing over all listed values,
taking expectation, and using ROI feasibility gives
\[
 \lambda W\leq\sum_{h\in\mathcal G}P_h\leq2W.
\]
Thus $\lambda\leq2\leq a$.

Now suppose that the guards have different multipliers.
Let $h^\ast$ be the guard with multiplier $\lambda$, and let
$\widehat\alpha<\lambda$ be the other guard's multiplier.
For every listed value $q$, the item of $h^\ast$ lies strictly
above the other guard's item of value $q$, whose bid is
$\widehat\alpha q$. Summing these payment bounds and using
ROI feasibility gives
\[
 \widehat\alpha W\leq P_{h^\ast}\leq W.
\]
Since $\widehat\alpha\geq1$, we must have $\widehat\alpha=1$.

Suppose, for a contradiction, that $\lambda>a$.
The item of $h^\ast$ with value $c^{2m+2}$ then bids strictly
above $ac^{2m+2}$. The other guard's item of value $ac^{2m+2}$
lies below it and forces a payment of at least
$ac^{2m+2}>c^{2m+2}$. Every other item of $h^\ast$ pays at
least its own value, because the other guard has an item of
that value and multiplier one. These bounds hold in every
allocation in the support of $\pi$, so $P_{h^\ast}>W$,
contradicting ROI feasibility. Hence $\lambda\leq a$.
\end{proof}

As in the budget case, this is where we use the bounded
forward span of the \PureCircuit\ instance. It ensures that
all items bidder $v$ receives as a gate input lie at levels
at most $\ell_v+D$, where $D=2d+1$ is a constant. Thus, even
if $v$ is used as an input by many gates,
\[
 V_v^{\mathrm{in}}\leq2a^{16}\bigl(1+c+\cdots+c^{\ell_v+D}\bigr)
 =O_\gamma(c^{\ell_v}).
\]
The implicit constant is independent of the circuit size,
which again allows us to choose $k$ independently of that size.
The next claim gives the bound in terms of $k$ and
$\varepsilon=1-\gamma$ that we will use below.

\begin{claim}
\label{clm:roi-tail}
For every real bidder $v\in V$,
\[
 V_v^{\mathrm{in}}\leq\frac{2a^{16}}{1-c^{-1}}c^{\ell_v+D}
 <\frac{3}{4a}kc^{\ell_v}
 \leq\frac{3\varepsilon}{8}kc^{\ell_v}.
\]
\end{claim}
\begin{proof}
Fix a real bidder $v$, and let $g_I$ be the unique gate of which
$v$ is an output. Its output level satisfies $\ell_v=2I$ or
$\ell_v=2I+1$.
Consider any gate $g_J$ that uses $v$ as an input. The forward-span
bound gives $J-I\leq d$ when $J>I$, and this inequality holds
automatically when $J<I$. Every item assigned to $v$ by this
gate lies at a level $\ell\leq2J+1$, so
\[
 \ell-\ell_v\leq2(J-I)+1\leq2d+1=D.
\]

At each occupied level $\ell$, bidder $v$ receives at most two
items as a gate input, even if it fills both input positions
of a $\NOR$ gate. Each such item has value at most $a^{16}c^\ell$.
Since all these items lie at integer levels at most $\ell_v+D$,
\[
 V_v^{\mathrm{in}}\leq2a^{16}\sum_{t=0}^{\ell_v+D}c^t
 <\frac{2a^{16}}{1-c^{-1}}c^{\ell_v+D}.
\]
Since $c=a^{128}$ and $a\geq2$, we have $2/(1-c^{-1})<3$.
Using $k=4a^{17}c^D$ and $1/a\leq\varepsilon/2$, we obtain
\[
 V_v^{\mathrm{in}}<3a^{16}c^{\ell_v+D}
 =\frac{3}{4a}kc^{\ell_v}
 \leq\frac{3\varepsilon}{8}kc^{\ell_v}.
\]
\end{proof}

The next lemma uses a guard item one level above each output
block to bound the corresponding bidder's multiplier.
The bound is strictly below the cap $A_0$, so approximate
maximal pacing will require every real bidder $v$ to pay at
least $(1-\gamma)V_v$ in expectation. Recall that
$\lambda\in[1,a]$ and $A_0>c>a^{65}$.

\begin{lemma}
\label{lem:roi-range}
Every real bidder $v\in V$ satisfies
\begin{equation}\label{eq:roi-range}
 \alpha_v a^{64}c^{\ell_v}\leq\lambda c^{\ell_v+1},
 \qquad 1\leq\alpha_v\leq\lambda a^{64}\leq a^{65}<A_0.
\end{equation}
Consequently, $P_v\geq(1-\gamma)V_v$.
\end{lemma}
\begin{proof}
Fix $v\in V$. Each output item has bid
$\alpha_v a^{64}c^{\ell_v}$. A guard with multiplier $\lambda$
has an item of value $c^{\ell_v+1}$, since
$\ell_v+1\leq2m+2$. If the output bid were strictly greater
than $\lambda c^{\ell_v+1}$, this guard item would lie below
every output item in every allocation in the support of $\pi$.
Each output item would therefore pay at least
$\lambda c^{\ell_v+1}$, giving
\[
 P_v=P_v^{\mathrm{out}}+P_v^{\mathrm{in}}
 \geq P_v^{\mathrm{out}}\geq k\lambda c^{\ell_v+1}\geq kc^{\ell_v+1}
 >ka^{64}c^{\ell_v}+V_v^{\mathrm{in}}=V_v,
\]
contrary to ROI feasibility. To see the strict inequality,
the difference between $kc^{\ell_v+1}$ and the output block's
value $ka^{64}c^{\ell_v}$ is $kc^{\ell_v}(c-a^{64})$.
Here $c-a^{64}=a^{64}(a^{64}-1)>1$, whereas
$V_v^{\mathrm{in}}<kc^{\ell_v}$ by \Cref{clm:roi-tail} and $\varepsilon\leq1$.

We have proved the first inequality in \eqref{eq:roi-range}.
Dividing by $a^{64}c^{\ell_v}$ gives
$\alpha_v\leq\lambda c/a^{64}=\lambda a^{64}\leq a^{65}$.
The lower bound $\alpha_v\geq1$ is part of the equilibrium
definition. Since $\alpha_v<A_0$, approximate maximal pacing
requires $P_v\geq(1-\gamma)V_v$.
\end{proof}

For the output block at an occupied level $\ell$, its
\emph{corresponding input items} are the items assigned to the
input bidders listed at that level in the construction.
Their owners differ from the output bidder, since a gate's
outputs are distinct from its inputs.

\begin{lemma}
\label{lem:roi-locality}
Fix an occupied level $\ell$, and let $v$ be its output bidder,
so $\ell_v=\ell$. The following hold:
\begin{enumerate}
\item Every input bid lies strictly between $\lambda c^\ell$
and $\lambda c^{\ell+1}$. The output bid exceeds
$\lambda c^\ell$ and is at most $\lambda c^{\ell+1}$.
\item If $\alpha_v<a^{64}$, then
$P_v^{\mathrm{in}}\leq a^{64}V_v^{\mathrm{in}}$.
\item The output items satisfy the following payment bounds:
\begin{enumerate}[label=(\roman*)]
\item If $\alpha_v<a^{64}$, every output item pays at most
the largest input bid.
\item If an input bid $b$ is strictly below the output bid,
every output item pays at least $b$.
\end{enumerate}
\end{enumerate}
\end{lemma}
\begin{proof}
Fix an allocation $\sigma$ in the support of $\pi$.
For (1), every input item at level $\ell$ has value either
$a^{16}c^\ell$ or $a^4c^\ell$. By \Cref{lem:roi-range},
its owner's multiplier belongs to $[1,a^{65}]$.
Its bid is therefore at least $a^4c^\ell>\lambda c^\ell$
and at most $a^{81}c^\ell<c^{\ell+1}\leq\lambda c^{\ell+1}$.
The output bid is at least $a^{64}c^\ell>\lambda c^\ell$,
and its upper bound follows from \Cref{lem:roi-range}.
These bounds hold at every occupied level. Thus real items
at lower levels have bids at most $\lambda c^\ell$, and real
items at higher levels have bids greater than
$\lambda c^{\ell+1}$.

For (2), every item pays at most its own bid. Summing over
the items bidder $v$ receives as a gate input gives a payment
of at most $\alpha_v V_v^{\mathrm{in}}\leq a^{64}V_v^{\mathrm{in}}$ whenever
$\alpha_v<a^{64}$. Taking expectation gives the stated bound.

For (3), the real items at level $\ell$ consist only of the
output block and its corresponding input items. The output
items all have the same owner, so they cannot set one another's
prices. If $\alpha_v<a^{64}$, the output bid is strictly below
$a^{128}c^\ell=c^{\ell+1}$. All higher-level real items and
all guard items of value at least $c^{\ell+1}$ lie above the
output block, since both guard multipliers are at least one.
Every remaining item outside this level bids
at most $\lambda c^\ell$, which is below every input bid.
Hence every output item pays at most the largest input bid,
proving (3)(i).

For (3)(ii), the input item with bid $b$ lies strictly below
every output item and has a different owner. Hence every
output item pays at least $b$. Both conclusions hold for every
$\sigma$ in the support of $\pi$.
\end{proof}

The strict condition $\alpha_v<a^{64}$ in part (3)(i) places
the output bid below $c^{\ell+1}$, and hence below both guard
bids at the next level. This is why the upper bound on output
payments does not require any assumption about tie-breaking.

\subsection{Decoding and proof of correctness}

For the fixed equilibrium $(\alpha,\pi)$, define an assignment
$\xx:V\to\{0,1,\bot\}$ using only the real bidders' multipliers:
\begin{equation}\label{eq:roi-decoding}
 \xx[v]=
 \begin{cases}
 0,&\alpha_v\leq a^{24},\\
 1,&\alpha_v\geq a^{64},\\
 \bot,&\text{otherwise}.
 \end{cases}
\end{equation}
Thus multipliers at most $a^{24}$ represent zero, multipliers
at least $a^{64}$ represent one, and intermediate multipliers
represent $\bot$. The two ranges are disjoint because
$a>1$. We show that this assignment satisfies every gate of
the given circuit.

As in the budget case, we first give two conditions on the
input bids at a level that force its output to be one or zero.
The low-input condition will follow from approximate maximal
pacing, and the high-input condition from ROI feasibility.

\begin{lemma}[Low input bids]
\label{lem:roi-low}
Fix an occupied level $\ell$, and let $v$ be its output bidder.
If every input bid at this level is at most $a^{60}c^\ell$,
then $\alpha_v\geq a^{64}$, and hence $\xx[v]=1$.
\end{lemma}
\begin{proof}
Suppose, for a contradiction, that $\alpha_v<a^{64}$.
Parts (3)(i) and (2) of \Cref{lem:roi-locality} give
$P_v^{\mathrm{out}}\leq ka^{60}c^\ell$ and
$P_v^{\mathrm{in}}\leq a^{64}V_v^{\mathrm{in}}$, respectively.
Recall that the output block has value $ka^{64}c^\ell$.
Its payment is therefore at most an $a^{-4}$ fraction of its
value. Since $a\geq2$ and $1/a\leq\varepsilon/2$, we have
$a^{-4}\leq\varepsilon/16$. Moreover, \Cref{clm:roi-tail}
gives $V_v^{\mathrm{in}}<(3\varepsilon/8)kc^\ell$.
Combining the two payment bounds yields
\[
 P_v=P_v^{\mathrm{out}}+P_v^{\mathrm{in}}
 \leq ka^{60}c^\ell+a^{64}V_v^{\mathrm{in}}
 <\frac{7\varepsilon}{16}ka^{64}c^\ell
 <(1-\gamma)V_v.
\]
The coefficient $7/16$ is the sum of $1/16$ and $3/8$, and
the last inequality uses $V_v\geq ka^{64}c^\ell$.
Thus the payments on input items cannot make up for the low
payment on the output block, contradicting \Cref{lem:roi-range}.
\end{proof}

\begin{lemma}[A high input bid]
\label{lem:roi-high}
Fix an occupied level $\ell$, and let $v$ be its output bidder.
If an input bid at this level is at least $a^{68}c^\ell$,
then $\alpha_v<a^{24}$, and hence $\xx[v]=0$.
\end{lemma}
\begin{proof}
Let $b$ be such an input bid. The output bid
$\alpha_v a^{64}c^\ell$ cannot be strictly greater than $b$:
otherwise part (3)(ii) of \Cref{lem:roi-locality} would make
every output item pay at least $b$, giving
\[
 P_v=P_v^{\mathrm{out}}+P_v^{\mathrm{in}}
 \geq P_v^{\mathrm{out}}\geq kb\geq ka^{68}c^\ell
 >ka^{64}c^\ell+V_v^{\mathrm{in}}=V_v.
\]
The strict inequality follows from
$a^{68}-a^{64}=a^{64}(a^4-1)>1$ and
$V_v^{\mathrm{in}}<kc^\ell$, where the latter bound is \Cref{clm:roi-tail}.
This contradicts ROI feasibility.

Thus $\alpha_v a^{64}c^\ell\leq b$.
The input item has value at most $a^{16}c^\ell$, and its
owner has multiplier at most $a^{65}$ by \Cref{lem:roi-range}.
Hence $b\leq a^{81}c^\ell$, and
\[
 \alpha_v\leq\frac{b}{a^{64}c^\ell}\leq a^{17}<a^{24}.
\]
\end{proof}

\begin{lemma}
\label{lem:roi-gates}
The assignment in \eqref{eq:roi-decoding} satisfies every gate
of the given \BFSPureCircuit$(5)$ instance.
\end{lemma}
\begin{proof}
We apply \Cref{lem:roi-low,lem:roi-high} at the levels of each gate.

\paragraph{$\NOT$ and $\NOR$ gates.}
Fix any $\NOT$ or $\NOR$ gate $g_I$ and set $\ell=2I$.
Every input item at this level has value $a^{16}c^\ell$.
If an input vertex is assigned zero, its multiplier is at
most $a^{24}$, so its input bid is at most
$a^{40}c^\ell\leq a^{60}c^\ell$.
Thus \Cref{lem:roi-low} forces output one for a $\NOT$ gate
with input zero and for a $\NOR$ gate with both inputs zero.

If an input vertex is assigned one, its multiplier is at
least $a^{64}$, so its input bid is at least
$a^{80}c^\ell\geq a^{68}c^\ell$.
By \Cref{lem:roi-high}, the output is assigned zero.
For a $\NOR$ gate, this holds whenever either input is one,
regardless of the other input. These are exactly the required
constraints.

\paragraph{$\NPURIFY$ gates.}
Fix any gate $g_I=\NPURIFY(u;v,w)$ and set $\ell=2I$.
The output blocks of $v$ and $w$ lie at levels $\ell$ and
$\ell+1$, respectively. The input items at these two levels
belong to $u$ and have bids
\[
 b^{(1)}=\alpha_u a^{16}c^\ell,
 \qquad b^{(2)}=\alpha_u a^4c^{\ell+1}.
\]
At the threshold $\alpha_u=a^{56}$, the first input bid is
$a^{72}c^\ell\geq a^{68}c^\ell$, while the second is
$a^{60}c^{\ell+1}$. Thus \Cref{lem:roi-high} forces the first
output to zero, while \Cref{lem:roi-low} forces the second
output to one. This allows the two conditions to cover every
possible input multiplier.

We first show that at least one output is pure.
If $\alpha_u\geq a^{56}$, then
$b^{(1)}\geq a^{72}c^\ell\geq a^{68}c^\ell$, so
\Cref{lem:roi-high} gives $\xx[v]=0$.
If $\alpha_u\leq a^{56}$, then
$b^{(2)}\leq a^{60}c^{\ell+1}$, so
\Cref{lem:roi-low} gives $\xx[w]=1$.
Every input multiplier satisfies at least one of these two
inequalities. At equality, both conclusions hold.

It remains to check the constraints when the input is pure.
If $\xx[u]=0$, then $\alpha_u\leq a^{24}<a^{56}$, so the
preceding argument gives $\xx[w]=1$. The first input bid
satisfies $b^{(1)}\leq a^{40}c^\ell\leq a^{60}c^\ell$,
so \Cref{lem:roi-low} also gives $\xx[v]=1$.
If $\xx[u]=1$, then $\alpha_u\geq a^{64}>a^{56}$, so the
preceding argument gives $\xx[v]=0$. The second input bid
satisfies $b^{(2)}\geq a^{68}c^{\ell+1}$, so
\Cref{lem:roi-high} also gives $\xx[w]=0$.
These are exactly the $\NPURIFY$ constraints.
\end{proof}

\begin{theorem}[Restatement of \Cref{thm:ppad-results}, ABE-R hardness]
\label{thm:roi-hardness}
For every fixed $0\leq\gamma<1$, computing a $\gamma$-approximate
ABE-R is \PPAD-hard even when
the number of slots equals the number of items and all
click-through rates equal one.
\end{theorem}
\begin{proof}
The construction has $kn$ items in the real bidders' output
blocks, $O(m)$ items assigned to gate inputs, and $2(2m+3)$
guard items. Since $k$ depends only on the fixed $\gamma$,
the total number of explicitly listed items is $O_\gamma(n+m)$.
Every item value is a power of the integer $a$ with a
nonnegative integer exponent $O(m)$, so it has $O_\gamma(m)$
bits. Computing the cap $A_0$ requires a sum and a division
of these values and preserves polynomial encoding length.
The auction can therefore be constructed in polynomial time.

For any $\gamma$-approximate ABE-R $(\alpha,\pi)$ of this
auction, \Cref{lem:roi-gates} shows that
\eqref{eq:roi-decoding} gives a solution of the source
instance. The decoding takes polynomial time and uses only
$\alpha$; the proof applies to every valid distribution $\pi$.
The claimed hardness follows from \Cref{thm:ordered-source}.
\end{proof}

\section{Conclusion}

We studied the computational complexity of autobidding equilibria
in GSP sponsored shopping auctions under both budget and ROI
constraints. For both settings, we establish \FIXP-membership
for exact equilibrium computation and \PPAD-membership for
inverse-exponential approximations. Our construction also
guarantees the existence of equilibria supported on polynomially
many deterministic allocations, in each of which items belonging
to the same bidder and having the same value are contiguous.
We obtain this guarantee by applying quantile coupling to legal
insertion marginals within PL pseudo-circuits. We expect this
approach to be useful in other equilibrium and allocation problems
where local choice distributions can be coupled into globally
feasible outcomes.

These structural guarantees also provide a basis for algorithm
development. In an earlier preprint version of this article,
we asked which structural restrictions permit efficient
equilibrium computation, in a similar spirit to the work
on exact algorithms for second-price pacing equilibria
in markets with few buyers or goods~\cite{YWL26,HYWL26}.
Very recent work~\cite{HLYW26}
established polynomial-time algorithms for computing exact
autobidding equilibria under both budget and ROI constraints
whenever either the number of bidders or the number of display
slots is fixed. Their algorithms use our existence guarantee
for equilibria in which items of the same bidder and value
form contiguous blocks, illustrating the algorithmic value
of this structural result.

On the hardness side, we establish tight inapproximability:
computing a $\gamma$-approximate equilibrium is \PPAD-hard
for every fixed $0\leq\gamma<1$, in both settings.
A remaining question is whether exact equilibrium computation
is \FIXP-hard, which, together with our membership results,
would establish \FIXP-completeness.

We would like to highlight a powerful methodological idea
behind our hardness results:
tailoring the structure of the source problem can unlock sharper
hardness guarantees. We introduce a new variant of \PureCircuit\
with bounded forward span and prove that it remains \PPAD-hard.
This structural restriction has a decisive quantitative payoff:
in our reductions, it bounds the total value of a bidder's input
items by a constant multiple of its output scale, independently
of the circuit size. This removes circuit size from the choice
of construction parameters, allowing them to depend only on
the approximation parameter and yielding hardness for every
fixed $\gamma<1$.

This suggests a broader program for strengthening \PPAD-hardness
results through the internal structure of \PureCircuit.
The guiding question is which structural properties of the
source circuit would make a reduction simpler or its guarantees
stronger. Establishing hardness under those restrictions can
then turn that structural advantage into an improved hardness
result for the target problem. Bounded forward span demonstrates
the potential of this approach; other restrictions, tailored
to the demands of other reductions, could open the way to
hardness guarantees that remain out of reach with existing
constructions.

\section*{AI methodology}
An earlier version (Version 1) of this paper, completed by June 2026, was developed
entirely by the authors without AI assistance. It established
\FIXP-membership and \PPAD-membership under the assumption that
each bidder's items have pairwise distinct values, together with
\PPAD-hardness for inverse-polynomial approximation parameters.
For the membership results, its allocation construction used doubly stochastic matrices to
obtain the support of randomized allocations and relied on the
distinct-value assumption.

During a second round of revisions around August 2026, discussions with ChatGPT-6 Astra
helped us recognize that quantile coupling could be incorporated
into our algebraic-circuit construction to remove this assumption.
We regard this connection as a nontrivial methodological insight. The resulting Version 2 established \FIXP-membership
and \PPAD-membership without the distinct-value assumption,
while retaining the inverse-polynomial hardness results.

In October 2026, the authors had an idea for a general strategy
to strengthen \PPAD-hardness results: examining the internal
structure of \PureCircuit\ rather than treating it as a black box.
We chose the present setting as a test
case because we were already working on it. We believe that
this strategy may apply more broadly to other \PPAD-hardness
reductions. Drawing on the
authors' reduction from Version 1, ChatGPT-6 Astra proposed a
variant of \PureCircuit\ with bounded forward span and developed
a proof of its \PPAD-hardness. Substituting this variant into
the existing reductions then yielded an almost immediate proof
of \PPAD-hardness for computing
$\gamma$-approximate ABE-B and ABE-R for
$0\leq\gamma\leq1/100$.
At this stage, however, ChatGPT-6 Astra was unable to extend
the hardness results substantially beyond $\gamma=1/100$.

The authors subsequently reviewed these proofs and found
simplifications. After incorporating these simplifications,
the authors further strengthened both hardness results through
discussions with ChatGPT-6 Astra, extending them to every fixed
$0\leq\gamma<1$ by parameterizing the constructions in terms
of $\gamma$.

The authors verified all the proofs and take full responsibility for all mathematical claims in the paper.
In the final revision, ChatGPT-6 Astra and Claude Fable 5.1 were also used to assist with drafting and revising the exposition and the \LaTeX\  source.

\section{Acknowledgements}
Kristoffer Arnsfelt Hansen and Xinhao Nie were supported by the Independent Research Fund Denmark under grant no.~10.46540/4283-00272B, and were supported in part by a grant of access to OpenAI models through the ChatGPT for Academic Researchers program. Yuhao Li was supported by NSF grants 2106429 and 2107187, and a grant from the AI Security Institute.

\bibliographystyle{alpha}
\bibliography{bib}

\newpage
\appendix

\section{Remarks on Multiplier Cap}
\label[appendix]{app:multiplier-cap}

The multiplier cap $A$ is used to give a bounded representation of an equilibrium, rather than as an additional economic parameter of the auction. This interpretation is already reflected in earlier work on ROI-constrained autobidding. Li and Tang explain that the cap makes the strategy space compact and observe that, for sufficiently large $A$, capped equilibria admit an equivalent uncapped interpretation in \cite{LT24}. Similarly, Dütting et al. \cite{DLPST26} take $A$ to be sufficiently large relative to the instance parameters. This qualitative convention is natural when studying equilibrium existence and efficiency. For the computational problem considered here, however, the phrase “sufficiently large” requires a more explicit formulation.

In particular, when the input consists of the auction itself, we must specify how a suitable $A$ is determined from the instance, bound its encoding length, and clarify how the output should be interpreted when $A$ is increased. Merely fixing an unspecified sufficiently large $A$ leaves implicit whether the task is to compute an equilibrium for one particular cap or to produce a single representation that remains valid for all sufficiently large caps. Our intended task is the latter, that is, the allocation distribution and all non-endpoint multipliers remain fixed, while multipliers at the upper endpoint are interpreted as the current cap.

Given an auction $G$, fix the sufficiently large $A$ as
\[
A_0 = 1 + \frac{\sum_{i,j} v_{i,j}}{\min_{i,j} v_{i,j}}.
\]
The computational task is to output an equilibrium $(\alpha^{(A_0)}, \pi)$ under $A_0$.
The solution $(\alpha^{(A_0)}, \pi)$ actually represents a family of equilibria $\{(\alpha^{(A)}, \pi)\}$ indexed by all $A \geq A_0$, where
\[
\alpha_i^{(A)} =
	\left\{ \begin{array}{ll}
	\alpha_i^{(A_0)} & \text{$\alpha_i^{(A_0)} < A_0$},\\
	A & \text{$\alpha_i^{(A_0)} = A_0$}.
	\end{array}
	\right.
\]

\begin{lemma}
    For each $A\in [A_0, \infty)$, $(\alpha^{(A)}, \pi)$ is an equilibrium with the same approximation parameters 
    if $(\alpha^{(A_0)}, \pi)$ is such an equilibrium. Thus, computing one equilibrium suffices to obtain 
    equilibria valid for all sufficiently large $A$.
\end{lemma}

\begin{proof}
We first show that, under $\mathbf b^{(A_0)}$, at most
one bidder can bid strictly greater than $\sum_{i,j} v_{ij}$.
Otherwise, in every allocation in the support of $\pi$,
the item occupying the first slot would have an item
of another bidder with bid greater than $\sum_{i,j} v_{ij}$ below it.
Its GSP payment would therefore exceed $c_1 \sum_{i,j}v_{ij}$.
Since all payments are nonnegative, this would imply
\[
\sum_{i\in[n]}P_i(\pi,\mathbf b^{(A_0)})
>c_1 \sum_{i,j} v_{ij}
\geq \sum_{i\in[n]}V_i(\pi),
\]
contradicting ROI feasibility.

Now fix $A\geq A_0$.
If no bidder has multiplier $A_0$, then
$\mathbf b^{(A)}=\mathbf b^{(A_0)}$ and the equilibrium is unchanged.
Otherwise, let $i^\star$ be a bidder whose multiplier
equals $A_0$. Every item of this bidder has bid at least
$A_0\min_{i,j}v_{ij}> \sum_{i,j} v_{ij}$. By the preceding argument, $i^\star$ is the unique bidder whose
multiplier equals $A_0$, and all its items precede all
other bidders' items in every allocation in the support
of $\pi$.
Increasing the multiplier of $i^\star$ to $A$ preserves
these orders, including the order among its own items.
Thus, $\pi$ remains bid-consistent.
The prices paid by $i^\star$ depend only on the unchanged bids of
other bidders. For every other bidder, all items of $i^\star$
precede its own items and hence cannot set their prices. Thus,
every bidder's value and payment remain unchanged in every
allocation in the support of $\pi$, and feasibility is preserved.
Bidder $i^\star$ is at the new cap. Every other bidder's multiplier
remains below the cap, and its exact or approximate maximal-pacing
condition continues to hold because its value and payment are
unchanged. Therefore, $(\alpha^{(A)},\pi)$ is an equilibrium with
the same approximation parameter.
\end{proof}

Our membership constructions apply directly 
with $A_0$ which is polynomial-time computable. It also imposes an additional requirement on our hardness reductions. 
In particular, we cannot rely on a prescribed
value of $A$ to calibrate the auction parameters or to separate bids associated with different gadgets. 
Instead, the multiplier bounds needed for such separation must be derived from the equilibrium conditions of the constructed 
auction itself. This rules out the convenient use of a fixed cap to control interactions between gadgets and makes the 
hardness proof technically more demanding.

\end{document}